\newcommand{\tAppA}{0}
\newcommand{\tAppB}{0}

\newif\ifnotes

\documentclass[table,dvipsnames,colorlinks=true]{lmcs}
\pdfoutput=1

\usepackage{lastpage}
\lmcsdoi{20}{2}{18}
\lmcsheading{}{\pageref{LastPage}}{}{}%
{Jul.~28,~2023}{Jun.~26,~2024}{}


\usepackage[utf8]{inputenc}
\usepackage[T1]{fontenc}
\usepackage{proof}

\ifnotes
  \usepackage[nomargin,inline,draft]{fixme}
\else
  \usepackage[nomargin,inline,final]{fixme}
\fi

\usepackage{xcolor}
\usepackage{amssymb}
\usepackage{amsmath}
\usepackage{url}
\usepackage{float}

\usepackage[
 colorlinks=true,
 citecolor=MidnightBlue,
 pdfauthor={Liron Cohen and Vincent Rahli},
 pdftitle={BoxTT},
 pdfsubject={BoxTT},
 pdfkeywords={Continuity, BoxTT, Agda},
]{hyperref}
\usepackage{cleveref}

\usepackage[anythingbreaks]{breakurl}

\theoremstyle{plain}
\newtheorem{remark}[thm]{Remark}
\newtheorem{example}[thm]{Example}

\usepackage[obeyFinal]{todonotes}

\usepackage{xspace}
\usepackage{stmaryrd}

\usepackage{fontawesome5}

\Crefname{equation}{Eq.}{Eqs.}
\Crefname{figure}{Fig.}{Figs.}
\Crefname{tabular}{Tab.}{Tabs.}
\Crefname{section}{Sec.}{Secs.}
\Crefname{definition}{Def.}{Defs.}
\Crefname{defi}{Def.}{Defs.}
\Crefname{lemma}{Lem.}{Lems.}
\Crefname{lem}{Lem.}{Lems.}
\Crefname{theorem}{Thm.}{Thms.}
\Crefname{thm}{Thm.}{Thms.}
\Crefname{paragraph}{Sec.}{Secs.}
\Crefname{appendix}{Appx.}{Appxs.}
\Crefname{corollary}{Cor.}{Cors.}
\Crefname{example}{Ex.}{Exs.}
\Crefname{proposition}{Prop.}{Props.}

\definecolor{darkraspberry}{rgb}{0.53, 0.15, 0.34}
\definecolor{pansypurple}{rgb}{0.47, 0.09, 0.29}
\definecolor{palatinatepurple}{rgb}{0.41, 0.16, 0.38}
\definecolor{ppp}{rgb}{0.5, 0.0, 0.5}

\newcommand{\formalised}{{\color{NavyBlue!75!White}{\raisebox{-0.5pt}{\scalebox{0.8}{\faCog}}}}}
\newcommand{\VGIT}{https://github.com/vrahli/opentt/blob/lmcs24}
\newcommand{\flink}[1]{\href{\VGIT/#1}{\formalised}}
\newcommand{\llink}[1]{\href{\VGIT/lmcs24.lagda\#L#1}{\formalised}}

\newcommand{\intitleb}[1]{{\medskip\noindent\textbf{\textsf{#1.}}}}

\newcommand{\Hcol}{ppp}%
\newcommand{\Hlink}[2]{\hyperref[#1]{{\color{\Hcol}{#2}}}}

\newcommand{\hide}[1]{}

\newcommand{\BM}{}%
\newcommand{\EM}{}%

\newcommand{\nat}{\mathbb{N}}

\newcommand{\META}[1]{\mathit{#1}}

\newcommand{\MEM}[1]{\mathsf{#1}}

\newcommand{\intitle}[1]{{\smallskip\noindent\textbf{\textsf{#1}}}}
\newcommand{\mytupleLEFT}{\langle}
\newcommand{\mytupleRIGHT}{\rangle}
\newcommand{\mytuple}[1]{\mytupleLEFT #1 \mytupleRIGHT}

\newcommand{\Bl}{\begin{tabular}[t]{@{}l@{}}}
\newcommand{\El}{\end{tabular}}
\newcommand{\Bi}{\begin{tabular}[t]{@{\quad}l}}
\newcommand{\Ei}{\end{tabular}}

\newcommand{\hidden}[1]{}

\newcommand{\NUPRL}[1]{\mathit{#1}}
\newcommand{\NUPRLC}[1]{\mathtt{#1}}
\newcommand{\NUPRLCS}[1]{{\boldsymbol{\mathtt{#1}}}}

\newcommand{\NproductSYMB}{\boldsymbol{\Pi}}
\newcommand{\Nproduct}[3]{\NproductSYMB#1\NUPRLC{:}#2.#3}

\newcommand{\NsumSYMB}{\boldsymbol{\Sigma}}
\newcommand{\Nsum}[3]{\NsumSYMB#1\NUPRLC{:}#2.#3}

\newcommand{\Nimplies}[2]{#1\rightarrow#2}

\newcommand{\Nfun}[2]{#2^{#1}}

\newcommand{\NUPRLarrow}[2]{#1\rightarrow#2}

\newcommand{\NUPRLfun}[2]{\NUPRLarrow{#1}{#2}}

\newcommand{\NUPRLsuba}[3]{#1[#2\backslash#3]}
\newcommand{\NUPRLsubb}[5]{#1[#2\backslash#3;#4\backslash#5]}

\newcommand{\NUPRLval}{\NUPRL{v}}
\newcommand{\NUPRLtval}{\NUPRL{vt}}

\newcommand{\NUPRLterm}{\NUPRL{t}}

\newcommand{\NUPRLnat}{{\mathbb N}}

\newcommand{\NUPRLuniverse}[1]{\mathbb{U}_{#1}}

\newcommand{\NUPRLlam}[2]{\lambda#1.#2}

\newcommand{\NUPRLfixSYMB}{\NUPRLC{fix}}
\newcommand{\NUPRLfix}[1]{\NUPRLfixSYMB(#1)}

\newcommand{\NUPRLspread}[4]{\NUPRLC{let}\ #2,#3=#1\ \NUPRLC{in}\ #4}

\newcommand{\NUPRLdecide}[5]{\NUPRLC{case}\ #1\ \NUPRLC{of}\ \NUPRLC{inl(}#2\NUPRLC{)}\Rightarrow#3\ \textbf{\texttt{|}}\ \NUPRLC{inr(}#4\NUPRLC{)}\Rightarrow#5}

\newcommand{\NUPRLnotSYMB}{\neg}
\newcommand{\NUPRLnot}[1]{\NUPRLnotSYMB#1}

\newcommand{\NUPRLapp}[2]{#1\ #2}
\newcommand{\NUPRLappb}[3]{#1\ #2\ #3}

\newcommand{\NUPRLapppar}[2]{#1(#2)}

\newcommand{\NUPRLequality}[3]{#2=#3\in#1}
\newcommand{\NUPRLequalityb}[3]{#2{=}#3\in#1}

\newcommand{\NUPRLinlSYMB}{\NUPRLC{inl}}
\newcommand{\NUPRLinrSYMB}{\NUPRLC{inr}}
\newcommand{\NUPRLinl}[1]{\NUPRLinlSYMB(#1)}
\newcommand{\NUPRLinr}[1]{\NUPRLinrSYMB(#1)}

\newcommand{\NUPRLsquashSYMB}{\NUPRLC{\downarrow}}
\newcommand{\NUPRLsquash}[1]{\NUPRLsquashSYMB#1}

\newcommand{\NUPRLunionSYMB}{\NUPRLC{+}}
\newcommand{\NUPRLunion}[2]{#1\NUPRLunionSYMB#2}

\newcommand{\NUPRLtrue}{\NUPRLC{Unit}}
\newcommand{\NUPRLfalse}{\NUPRLC{Void}}

\newcommand{\NUPRLpair}[2]{\langle#1,#2\rangle}

\newcommand{\NUPRLvoid}{\NUPRLC{Void}}

\newcommand{\NUPRLmember}[2]{#2\in#1}

\newcommand{\NUPRLset}[3]{\{#1:#2\mid#3\}}

\newcommand{\NUPRLhyp}[2]{#1:#2}

\newcommand{\NUPRLhhyp}[2]{[#1:#2]}

\newcommand{\NUPRLmetaand}[2]{#1\wedge#2}
\newcommand{\NUPRLmetaandc}[3]{#1\wedge#2\wedge#3}

\newcommand{\NUPRLcomputestoSYMB}{\Downarrow}

\newcommand{\NUPRLcomputestop}[3]{#2\Downarrow_{#1}#3}

\newcommand{\NUPRLmetainsymb}{\NUPRLCS{\in}}

\newcommand{\Ncantor}{{\cal C}}
\newcommand{\Nbaire}{{\cal B}}

\newcommand{\Nqsquash}[1]{\NqsquashSYMB#1}

\newcommand{\NstepSYMB}{\mapsto}
\newcommand{\NstepwwSYMB}[2]{\NstepSYMB^{#1}_{#2}}

\newcommand{\Nstepwwa}[4]{#1\NstepwwSYMB{#3}{#4}#2}

\newcommand{\Naxiom}{\NUPRLC{\star}}

\newcommand{\NfreshSYMB}{\NUPRLCS{\nu}}

\newcommand{\Nfresh}[2]{\NfreshSYMB#1.#2}

\definecolor{glink}{gray}{0.2}
\newcommand{\myurl}[1]{{\small\textcolor{RoyalBlue}{\url{#1}}}}

\newcommand{\NreadSYMB}{!}
\newcommand{\Nread}[1]{\NreadSYMB#1}

\newcommand{\REF}{\Hlink{ex:ref}{\MEM{Ref}}}

\newcommand{\bool}{\mathbb{B}}
\newcommand{\Mtt}{\MEM{true}}
\newcommand{\Mff}{\MEM{false}}
\newcommand{\contp}{\MEM{CONT_p}}

\renewcommand{\NUPRLtrue}{\MEM{Unit}}%
\renewcommand{\NUPRLfalse}{\MEM{Void}}%
\newcommand{\Nbool}{\MEM{Bool}}%

\renewcommand{\NUPRLnat}{\MEM{Nat}}

\newcommand{\wcp}{\texttt{WCP}}

\newcommand{\NsquashSYMB}{\downarrow}
\newcommand{\Nsquash}[1]{{\NsquashSYMB}(#1)}
\newcommand{\NassertSYMB}{\uparrow}
\newcommand{\Nassert}[1]{{\NassertSYMB}(#1)}
\newcommand{\Nassertb}[1]{{\NassertSYMB}#1}
\newcommand{\Nbtrue}{\MEM{tt}}%
\newcommand{\Nbfalse}{\MEM{ff}}%

\newcommand{\MfamSYMB}{\Hlink{sec:forcing}{\MEM{Fam}}}
\newcommand{\Mfam}[5]{\MfamSYMB_{#1}(#2,#3,#4,#5)}

\renewcommand{\myurl}[1]{{\footnotesize\textcolor{RoyalBlue}{\url{#1}}}}

\newcommand{\NpureproductSYMB}{\boldsymbol{\Pi_{p}}}
\newcommand{\Npureproduct}[3]{\NpureproductSYMB#1\NUPRLC{:}#2.#3}

\newcommand{\NpureSYMB}{\NUPRLC{Pure}}
\newcommand{\Npure}{\NpureSYMB}

\newcommand{\NisectSYMB}{\cap}
\newcommand{\Nisect}[2]{#1\NisectSYMB#2}

\newcommand{\NnonamesSYMB}{\NUPRLC{namefree}}
\newcommand{\Nnonames}[1]{\NnonamesSYMB(#1)}

\newcommand{\NnoreadsSYMB}{\NUPRLC{readfree}}
\newcommand{\Nnoreads}[1]{\NnoreadsSYMB(#1)}

\newcommand{\NnowritesSYMB}{\NUPRLC{writefree}}
\newcommand{\Nnowrites}[1]{\NnowritesSYMB(#1)}

\newcommand{\NnoreadSYMB}{\NUPRLC{NoRead}}
\newcommand{\NnowriteSYMB}{\NUPRLC{NoWrite}}
\newcommand{\Nnoreadmod}[1]{|#1|_{{\mathtt{r}}}}
\newcommand{\Nnowritemod}[1]{|#1|_{{\mathtt{w}}}}
\newcommand{\Npuremod}[1]{|#1|_{{\mathtt{c}}}}
\newcommand{\Nnoreadwritemod}[1]{|#1|_{\mathtt{r}\mathtt{w}}}

\newcommand{\NUPRLnatnrw}{\Nnoreadwritemod{\NUPRLnat}}
\newcommand{\Nboolnrw}{\Nnoreadwritemod{\Nbool}}
\newcommand{\NUPRLnatnr}{\Nnoreadmod{\NUPRLnat}}
\newcommand{\Nboolnr}{\Nnoreadmod{\Nbool}}
\newcommand{\NUPRLnatr}{\MEM{N}}

\renewcommand{\Nbaire}{{\mathfrak B}}
\renewcommand{\Ncantor}{{\mathfrak C}}

\newcommand{\NTval}{\Hlink{fig:syntax}{\MEM{Value}}}
\newcommand{\NTtype}{\Hlink{fig:syntax}{\MEM{Type}}}
\newcommand{\NTterm}{\Hlink{fig:syntax}{\MEM{Term}}}
\newcommand{\NTvar}{\MEM{Var}}

\newcommand{\NUPRLcomputestoallSYMB}{\Rrightarrow}%
\newcommand{\NUPRLcomputestoall}[3]{#2\NUPRLcomputestoallSYMB_{#1}#3}
\newcommand{\NUPRLcomputestoallbSYMB}[1]{\NUPRLcomputestoallSYMB_{!#1}}
\newcommand{\NUPRLcomputestoallb}[3]{#2\NUPRLcomputestoallbSYMB{#1}#3}

\newcommand{\metanat}[1]{\underline{#1}}

\newcommand{\Nqbaire}[1]{\Nbaire^{r}_{#1}}
\newcommand{\Nqlt}[2]{#1<^{r}#2}

\newcommand{\metatrue}{\NUPRLC{True}}

\newcommand{\myfbox}[1]{\fcolorbox{gray}{gray!20}{\mbox{#1}}}%
\newcommand{\hidebox}[1]{\fcolorbox{blue}{blue!8}{\color{blue}#1}}
\newcommand{\shidebox}[1]{{\tiny\hidebox{#1}}}

\renewcommand{\NUPRLhhyp}[2]{\fbox{$#1{:}#2$}}
\renewcommand{\NUPRLhyp}[2]{#1{:}#2}

\newcommand{\DEF}{:\mkern-2mu\equiv}

\newcommand{\Mparam}[2]{(#1:#2)}

\newcommand{\Mall}[3]{\forall\Mparam{#1}{#2}.#3}
\newcommand{\Mallp}[2]{\forall#1.#2}

\newcommand{\Mallb}[5]{\forall(#1:#2)(#3:#4).#5}

\newcommand{\Mexp}[2]{\exists#1.#2}

\newcommand{\Mfun}[2]{#1\rightarrow#2}

\newcommand{\Mor}[2]{#1\vee#2}

\newcommand{\Mand}[2]{#1\wedge#2}

\newcommand{\Mandc}[3]{#1\wedge#2\wedge#3}

\newcommand{\Mandd}[4]{#1\wedge#2\wedge#3\wedge#4}

\newcommand{\Mworld}{\Hlink{def:worlds}{\mathcal{W}}}
\newcommand{\Mwpred}[1]{{\mathcal{P}}_{#1}}

\newcommand{\Mres}{\Hlink{def:restrictions}{\MEM{Res}}}
\newcommand{\MdresSYMB}{\Hlink{def:restrictions}{\cdot\NUPRLC{d}}}
\newcommand{\Mdres}[1]{#1_{\MdresSYMB}}
\newcommand{\Mr}{\META{r}}
\newcommand{\Mrc}{\NUPRLC{r}}

\newcommand{\Mlam}[2]{\lambda#1.#2}
\newcommand{\Mexw}[2]{\exists^{\sqsubseteq}_{#1}(#2)}
\newcommand{\Mexww}[3]{\exists^{\sqsubseteq}_{#1}(#2.#3)}
\newcommand{\Mallw}[2]{\forall^{\sqsubseteq}_{#1}(#2)}
\newcommand{\Mallww}[3]{\forall^{\sqsubseteq}_{#1}(#2.#3)}
\newcommand{\Mallwwp}[3]{\forall^{\sqsubseteq}_{#1}\left(#2.#3\right)}

\newcommand{\MextSYMB}{\Hlink{def:worlds}{\sqsubseteq}}
\newcommand{\MextpSYMB}{\Hlink{def:worlds}{\sqsupseteq}}
\newcommand{\Mext}[2]{#1\mathrel{\MextSYMB}#2}
\newcommand{\Mextp}[2]{#1\mathrel{\MextpSYMB}#2}
\newcommand{\Mprop}{{\mathbb P}}

\newcommand{\Mw}{\META{w}}
\newcommand{\MbarmodSYMB}{\Box}
\newcommand{\MbarmodM}[1]{\MbarmodSYMB_{#1}}
\newcommand{\Mbarmod}[2]{\MbarmodM{#1}#2}
\newcommand{\Mbarmodw}[3]{\MbarmodM{#1}(#2.#3)}

\newcommand{\Mwplift}[2]{#2} %
\newcommand{\MchainSYMB}{\Hlink{def:chains}{\MEM{chain}}}
\newcommand{\Mchain}[1]{\MchainSYMB(#1)}
\renewcommand{\Mc}{\META{c}}

\newcommand{\Meqtypes}[3]{#1\vDash#2\NUPRLCS{\scriptstyle\equiv}#3}
\newcommand{\Mistype}[2]{#1\vDash\NUPRLC{type}(#2)}
\newcommand{\Mintype}[4]{#1\vDash#2\NUPRLCS{\scriptstyle\equiv}#3\NUPRLmetainsymb#4}
\newcommand{\Mmemtype}[3]{#1\vDash#2\NUPRLmetainsymb#3}
\newcommand{\Minhtype}[2]{#1\vDash#2}

\newcommand{\McSYMB}{{\mathcal C}}

\newcommand{\MnewChoiceSYMB}{\Hlink{def:new-choices}{\MEM{\nu}\McSYMB}}%
\newcommand{\MnewChoice}[1]{\MnewChoiceSYMB(#1)}

\newcommand{\MstartNewChoiceSYMB}{\Hlink{def:new-choices}{\MEM{start\nu}\McSYMB}}
\newcommand{\MstartNewChoice}[2]{\MstartNewChoiceSYMB(#1,#2)}

\newcommand{\set}{set\xspace}

\newcommand{\Mca}{\Hlink{def:choices}{\boldsymbol{\kappa_0}}}
\newcommand{\Mcb}{\Hlink{def:choices}{\boldsymbol{\kappa_1}}}
\newcommand{\Mcn}{\delta}%
\newcommand{\Mcname}{\Hlink{def:choice-name}{{\mathcal N}}}
\newcommand{\Mctype}{\Hlink{def:choices}{\McSYMB}}

\newcommand{\Mctot}[1]{#1}%
\newcommand{\MgetchoiceSYMB}{\Hlink{def:choices}{\MEM{read}}}
\newcommand{\Mgetchoice}[2]{\MgetchoiceSYMB(#1,#2)}
\newcommand{\Mgetchoiceb}[3]{\MgetchoiceSYMB(#1,#2,#3)}
\newcommand{\MchooseSYMB}{\Hlink{def:choosing}{\MEM{write}}}
\newcommand{\Mchoose}[3]{\MchooseSYMB(#1,#2,#3)}

\newcommand{\Mapp}[2]{#1\ #2}
\newcommand{\Mapppar}[2]{#1(#2)}
\newcommand{\McompatibleSYMB}{\Hlink{def:compatibility}{\MEM{comp}}}
\newcommand{\Mcompatible}[3]{\McompatibleSYMB(#1,#2,#3)}
\newcommand{\Mch}{c}
\newcommand{\modaa}{\Hlink{def:mod-props}{\MbarmodSYMB_{\boldsymbol{1}}}}
\newcommand{\modab}{\Hlink{def:mod-props}{\MbarmodSYMB_{\boldsymbol{2}}}}
\newcommand{\modad}{\Hlink{def:mod-props}{\MbarmodSYMB_{\boldsymbol{3}}}}
\newcommand{\modae}{\Hlink{def:mod-props}{\MbarmodSYMB_{\boldsymbol{4}}}}
\newcommand{\modaf}{\Hlink{def:mod-props}{\MbarmodSYMB_{\boldsymbol{5}}}}

\renewcommand{\NUPRLequality}[3]{#2\NUPRLC{=}#3\NUPRLC{\in}#1}
\renewcommand{\NUPRLequalityb}[3]{(\NUPRLequality{#1}{#2}{#3})}
\renewcommand{\NUPRLmember}[2]{#2\NUPRLC{\in}#1}

\newcommand{\NchooseSYMB}{\coloneqq}%
\newcommand{\Nchoose}[2]{#1\mathrel{\NchooseSYMB}#2}%
\newcommand{\Nsq}[2]{#1\NUPRLC{;}#2}
\newcommand{\NupdSYMB}{\Hlink{def:continuity+realizer}{\NUPRLC{upd}}}
\newcommand{\Nupd}[2]{\NupdSYMB(#1,#2)}
\newcommand{\NforceSYMB}{\Hlink{def:force-cbv}{\NUPRLC{cbv}}}
\newcommand{\Nforce}[1]{\NforceSYMB(#1)}
\newcommand{\Nlet}[3]{\NUPRLC{let}\ #1\,\,{=}\,\,#2\ \NUPRLC{in}\ #3}
\newcommand{\Nlt}[2]{#1\ <\ #2}
\newcommand{\Nite}[3]{\NUPRLC{if}\ #1\ \NUPRLC{then}\ #2\ \NUPRLC{else}\ #3}
\newcommand{\Niflt}[4]{\Nite{\Nlt{#1}{#2}}{#3}{#4}}

\newcommand{\Nsub}[2]{#1-#2}
\newcommand{\NpredSYMB}{\MEM{pred}}
\newcommand{\Npred}[1]{\NpredSYMB(#1)}
\newcommand{\NsuccSYMB}{\MEM{succ}}
\newcommand{\Nsucc}[1]{\NsuccSYMB(#1)}
\newcommand{\NnegSYMB}{\MEM{neg}}
\newcommand{\Nneg}[1]{\NnegSYMB(#1)}
\newcommand{\NiszeroSYMB}{\MEM{iszero}}
\newcommand{\Niszero}[1]{\NiszeroSYMB(#1)}
\newcommand{\NmodSYMB}{\Hlink{def:continuity+realizer}{\NUPRLC{mod}}}
\newcommand{\Nmod}[2]{\NmodSYMB(#1,#2)}

\newcommand{\Mba}{\alpha}
\newcommand{\Mbb}{\beta}

\newcommand{\bctt}[2]{\text{TT}^{#1}_{#2}\xspace}
\newcommand{\Btt}{$\bctt{\MbarmodSYMB}{\McSYMB}$\xspace}

\newcommand{\Nsequentext}[3]{#1\vdash#3:#2}
\newcommand{\Nsequent}[2]{#1\vdash#2}

\newcommand{\Pextendable}{\Hlink{def:new-choices}{extendable}\xspace}

\newcommand{\Pextendability}{\Hlink{def:new-choices}{extendability}\xspace}

\newcommand{\PExtendability}{Extendability\xspace}

\newcommand{\PMutability}{Mutability\xspace}
\newcommand{\Pmutable}{\Hlink{def:choosing}{mutable}\xspace}

\newcommand{\Pmutability}{\Hlink{def:choosing}{mutability}\xspace}

\newcommand{\PCompatibility}{Compatibility\xspace}
\newcommand{\Pcompatible}{\Hlink{def:compatibility}{compatible}\xspace}

\newcommand{\Pnontrivial}{\Hlink{def:non-trivial-choices}{non-trivial}\xspace}

\renewcommand{\NUPRLcomputestoSYMB}{\Mapsto}
\renewcommand{\NUPRLcomputestop}[3]{#2\NUPRLcomputestoSYMB_{#1}#3}
\newcommand{\NUPRLcomputestopp}[4]{#3\NUPRLcomputestoSYMB^{#1}_{#2}#4}

\newcommand{\MsimdiffSYMB}{\Hlink{def:sim1}{\sim}}
\newcommand{\Msimdiff}[5]{#1\mathrel{\MsimdiffSYMB_{#3,#4,#5}}#2}
\newcommand{\MsimforceSYMB}{\Hlink{def:sim2}{\approx}}
\newcommand{\Msimforce}[5]{#1\mathrel{\MsimforceSYMB_{#3,#4,#5}}#2}
\newcommand{\MupdtermSYMB}{\Hlink{def:updterm}{\NUPRLC{Upd}}}
\newcommand{\Mupdterm}[3]{\MupdtermSYMB_{#2,#3}(#1)}

\newcommand{\assa}{\Hlink{sec:assumptions}{\text{Asm}_1}}
\newcommand{\assb}{\Hlink{sec:assumptions}{\text{Asm}_2}}
\newcommand{\assc}{\Hlink{sec:assumptions}{\text{Asm}_3}}

\renewcommand{\Nqsquash}[1]{\|#1\|}

\newcommand{\NnatrecSYMB}{\MEM{natrec}}
\newcommand{\Nnatrec}[3]{\NnatrecSYMB(#1,#2,#3)}

\newcommand{\McovSYMB}{\triangleleft}
\newcommand{\Mcov}[2]{#1\McovSYMB#2}
\newcommand{\Mcovn}[3]{#2\McovSYMB_{#1}#3}
\newcommand{\Mopen}{o}

\fxusetheme{color}
\fxuseenvlayout{color}

\definecolor{vrpink}{RGB}{255,0,127}
\definecolor{vrblue}{RGB}{30,144,255}
\definecolor{vrolive}{RGB}{85,107,47}
\definecolor{vrroyalblue}{RGB}{65,105,225}
\definecolor{vrlpink}{RGB}{255,192,203}

\definecolor{dgreen}{RGB}{0,100,0}
\definecolor{dred}{RGB}{139,0,0}

\FXRegisterAuthor{vr}{avr}{\color{vrpink}[vince]}
\FXRegisterAuthor{lc}{alc}{\color{vrroyalblue}[liron]}

\floatstyle{ruled}
\restylefloat{figure}

\begin{document}

\title[\Btt: a Family of Effectful, Extensional Type Theories]{\Btt: a Family of Extensional Type Theories with Effectful Realizers of Continuity}

\thanks{This research was partially supported by Grant No. 2020145
  from the United States-Israel Binational Science Foundation (BSF)}

\author[L.~Cohen]{Liron Cohen\lmcsorcid{0000-0002-6608-3000}}[a]
\author[V.~Rahli]{Vincent Rahli\lmcsorcid{0000-0002-5914-8224}}[b]

\address{Ben-Gurion University, Israel}
\email{cliron@cs.bgu.ac.il}

\address{University of Birmingham, UK}
\email{V.Rahli@bham.ac.uk}

\begin{abstract}
\Btt is a generic family of effectful, extensional type theories with a forcing interpretation parameterized by modalities. 
  This paper identifies a subclass of \Btt theories that internally realizes continuity principles through stateful computations, such as reference cells.
  The principle of continuity is a seminal property that holds for a
  number of intuitionistic theories such as System T. 
  Roughly speaking, it states that functions on real numbers only need
  approximations of these numbers to compute. 
  Generally, continuity
  principles have been justified using semantical arguments, but it is
  known that the modulus of continuity of functions can be computed
  using effectful computations such as exceptions or reference cells.
In this paper, the modulus of
  continuity of the functionals on the Baire space is directly
  computed using the stateful computations enabled internally in the theory.
\end{abstract}

\maketitle

\hide{
}

\section{Introduction}
\label{sec:introduction}

The framework \Btt~\cite{Cohen+Rahli:fscd:2022} is a generic family of
effectful, extensional type theories with a forcing interpretation
parameterized by modalities. More concretely, \Btt uses a general
possible-worlds forcing interpretation parameterized by an abstract
modality $\MbarmodSYMB$, which, in turn, can be instantiated with
simple covering relations, leading to a general sheaf model. In
addition, \Btt is parameterized by a type of time-progressing choice
operators ~$\Mctype$, enabling support for internal effectful
computations.
\Btt is particularly suitable for studying effectful theories, and
indeed, \Btt was called an ``unprejudiced'' type theory since these
parameters can be instantiated to obtain theories that are either
``agnostic'', i.e., compatible with classical reasoning (in the sense
that classical axioms, such as the Law of Excluded Middle, can be
validated), or that are ``intuitionistic'', i.e., incompatible with
classical reasoning (in the sense that classical axioms can be proven
false).

This paper uses the \Btt framework to reason about the continuity
principle which is a seminal property in intuitionistic theories which
contradicts classical mathematics but is generally accepted by
constructivists.
Roughly speaking, the principle states that functions on real numbers
only need approximations of these numbers to compute.
Brouwer, in particular, assumed his so-called \emph{continuity
principle for numbers} to derive that all real-valued functions on the
unit interval are uniformly
continuous~\cite{Kleene+Vesley:1965,Dummett:1977,Beeson:1985,Bridges+Richman:1987,Troelstra+VanDalen:1988}.
The continuity principle for numbers, sometimes referred to as the
weak continuity principle, states that all functions on the Baire
space (i.e., $\Nbaire\DEF\NUPRLfun{\NUPRLnat}{\NUPRLnat}$, the set of
infinite sequences of numbers) have a modulus of continuity. More
concretely, given a function $F$ of type
$\NUPRLfun{\Nbaire}{\NUPRLnat}$ and a function $\alpha$ of
type~$\Nbaire$, the principle states that $\NUPRLapppar{F}{\alpha}$
can only depend on an initial segment of $\alpha$, and the length of
the smallest such segment is the modulus of continuity of $F$ at
$\alpha$. This is standardly formalized as follows, where
$\Nbaire_{n}\DEF\NUPRLfun{\NUPRLset{x}{\NUPRLnat}{x<n}}{\NUPRLnat}$ is
the set of finite sequences of numbers of length $n$:
$$\wcp
=\Nproduct
{F}
{\NUPRLfun{\Nbaire}{\NUPRLnat}}
{\Nproduct
  {\alpha}
  {\Nbaire}
  {\Nqsquash
    {\Nsum
      {n}
      {\NUPRLnat}
      {\Nproduct
        {\beta}
        {\Nbaire}
        {\Nimplies
          {(\NUPRLequality{\Nbaire_n}{\alpha}{\beta})}
          {(\NUPRLequality{\NUPRLnat}{\NUPRLapppar{F}{\alpha}}{\NUPRLapppar{F}{\beta}})}
        }
      }
    }
  }
}$$

A number of theories have been shown to satisfy Brouwer's continuity
principle, or uniform variants, such as N-HA$^{\omega}$ by
Troelstra~\cite[p.158]{Troelstra:1973}, MLTT by Coquand and
Jaber~\cite{Coquand+Jaber:2010,Coquand+Jaber:2012}, System~T by
Escard\'o~\cite{Escardo:2013}, CTT by Rahli and
Bickford~\cite{Rahli+Bickford:cpp:2016}, BTT by Baillon, Mahboubi and
Pedrot~\cite{Baillon+Mahboubi+Pedrot:csl:2022}, to cite only a few
(see \Cref{sec:related-work} for further details).
These proofs often rely on a semantical forcing-based
approach~\cite{Coquand+Jaber:2010,Coquand+Jaber:2012}, where the
forcing conditions capture the amount of information needed when
applying a function to a sequence in the Baire space, or through
suitable models that internalize (C-Spaces in~\cite{Xu+Escardo:2013})
or exhibit continuous behavior (e.g., dialogue trees
in~\cite{Escardo:2013,Baillon+Mahboubi+Pedrot:csl:2022}).

Not only can functions on the Baire space be proved to be continuous,
but using effectful computations, as for example described
in~\cite{Longley:1999}, one can \emph{compute} the modulus of continuity of
such a function.
However, as shown for example by Kreisel~\cite[p.154]{Kreisel:1962},
Troelstra~\cite[Thm.IIA]{Troelstra:1977b}, and Escard\'o and
Xu~\cite{Escardo+Xu:2015,Xu:phd:2015}, continuity is not an
extensional property in the sense that two equal functions might have
different moduli of continuity.
Therefore, to realize continuity, the existence of a modulus of continuity has to be
truncated as explained, e.g.,
in~\cite{Escardo+Xu:2015,Xu:phd:2015,Rahli+Bickford:cpp:2016,Rahli+Bickford:mscs:2017},
which is what the $\Nqsquash{\_}$ operator achieves in $\wcp$.
Following the effectful approach, continuity was shown to be realizable
in~\cite{Rahli+Bickford:cpp:2016,Rahli+Bickford:mscs:2017} using
exceptions.

Instead of using exceptions, a more straightforward way to compute
the modulus of continuity of a function on the Baire space is to use
reference cells. This was explained, e.g., in~\cite{Longley:1999},
where the use of references can be seen as the programming
counterparts of the more logical forcing conditions.
The computation using references is more efficient than when
  using exceptions as it allows computing the modulus of continuity of
  a function $F$ at a point $\alpha$ simply by executing $F$ on
  $\alpha$, while recording the highest argument that $\alpha$ is
  applied to, while using exceptions requires
  repeatedly searching for the modulus of continuity.

Following this line of work, in this paper we show how to use stateful
computations to realize a continuity principle. This allows deriving
constructive type theories that include continuity axioms where the
modulus of continuity is internalized in the sense that it is computed
by an expression of the underlying programming language.
Concretely, we do so for \Btt, which is
presented in more details in \Cref{sec:background}.
More precisely, we prove in this paper that all \Btt functions are
continuous for some instances of~$\MbarmodSYMB$ and~$\Mctype$: namely
for ``non-empty'' equality modalities, and reference-like stateful
choice operators.
Our proof is for a variant of the weak continuity principle (see
\Cref{def:continuity+realizer}), which we show to be inhabited by a
program that relies on a choice operator to keep track of the
modulus of continuity of a given function, following Longley's
method~\cite{Longley:1999}.
This variant is restricted to ``pure'' (i.e., without side effects)
functions $F$, $\alpha$, and $\beta$, and \Cref{sec:purity} discusses
issues arising with impure functions.%

\intitle{Roadmap.}
After presenting in \Cref{sec:background} the main aspects of \Btt that
are relevant to the results presented in this paper,
\Cref{sec:continuity-proof} validates a continuity principle using
stateful computations.
One key contribution of this paper, discussed in~\Cref{sec:btt-inst},
is the fact that \Btt allows computations to modify the current world,
which is accounted for in its forcing interpretation.
Some consequences of this fact are further discussed in
\Cref{sec:principles}.
Another key contribution, discussed in \Cref{sec:continuity-proof}, is
the internalization of the modulus of continuity of functions, in the
sense that it can be computed by a \Btt expression and used to
validate the continuity principle.
Finally, \Cref{sec:related-work} concludes and discusses the related work on continuity.

\section{\texorpdfstring{\Btt}{BoxTT} : Syntax \& Semantics}
\label{sec:background}
\label{sec:btt-inst}

This section presents \Btt, a family of type theories introduced
in~\cite{Cohen+Rahli:fscd:2022}, which is parameterized by a choice
operator~$\Mctype$, and a metatheoretical modality $\MbarmodSYMB$,
which allows typing the choice operator.
The choice operators are time-progressing elements that we will in
particular instantiate with references.
\Cref{sec:btt-inst} carves out a sub-family for which we can validate
computationally relevant continuity rules as shown in
\Cref{sec:continuity-proof}.

The version presented here extends the one introduced
in~\cite{Cohen+Rahli:fscd:2022} in particular with the following
components, which are formally defined next.%
\begin{itemize}
\item An operator that allows making a choice
  ($\Nchoose{t_1}{t_2}$). Computations are performed against worlds
  (see \Cref{sec:worlds}), and~\cite{Cohen+Rahli:fscd:2022} already
  provided computations to ``read'' choices from a world. However,
  even though~\cite{Cohen+Rahli:fscd:2022} included metatheoretical
  computations to update a world in the form of the \Pmutability
  requirement presented in \cref{def:choosing}, it did not include
  corresponding object computations. Doing so has far-reaching
  consequences. We generalize \Btt's semantics accordingly, in effect
  internalizing the \Pmutability requirement.
\item An operator to generate a ``fresh'' choice name
  ($\Nfresh{x}{t}$). While the version of \Btt presented
  in~\cite{Cohen+Rahli:fscd:2022} provided metatheoretical
  computations to generate new choice names in the form of the
  \Pextendability requirement presented in \Cref{def:new-choices}, it
  did not provide corresponding object computations. We remedy this
  here by extending \Btt with a corresponding computation, which
  internalizes the \Pextendability requirement.%
\item A type that states the ``purity'' of an expression, i.e., that
  the expression has no side effects. This will allow us to formalize
  the variant of the continuity principle we validate.
  \Cref{sec:purity} provides further details.
\end{itemize}
Moreover, the version presented here differs from the one presented
in~\cite{Cohen+Rahli:csl:2023} as follows:%
\begin{itemize}
\item \Cref{sec:background} contains further details regarding \Btt,
  such as examples illustrating how effectful programs behave and are
  given meaning through \Btt types, as well as a discussion of \Btt's
  inference rules.
\item The way \Btt captures effects is simpler and more uniform. The
  type theory is more uniform in the sense that types are now impure
  by default and \Btt provides modalities to capture different levels
  of purity, as opposed to~\cite{Cohen+Rahli:csl:2023} where types
  were what is characterized as ``write-only'' in
  \Cref{sec:different-kinds-of-effects}, and the theory provided
  modalities to both make types more pure and more impure (through a
  complex ``time truncation'' type operator).
  Furthermore, the semantics of these new modalities (see
  \Cref{sec:forcing}) is simpler compared to the semantics of the
  ``time truncation'' operator used in~\cite{Cohen+Rahli:csl:2023},
  which was used to turn a ``write-only'' type into a ``read \&
  write'' type (see \Cref{sec:different-kinds-of-effects}).
\item \Cref{sec:principles} discusses how extending \Btt with
  computations to update the current world impacts validity results
  presented in~\cite{Cohen+Rahli:fscd:2022} of standard axioms such as
  Markov's Principle.
\end{itemize}

\subsection{Metatheory}
Our metatheory is Agda's type theory~\cite{Agda}.
The results presented in this paper have been formalized in Agda, and
the formalization is available
here:~\myurl{\VGIT}. We will use the
symbol \formalised{} to link to the corresponding definition or result
in the formalization.
We use $\forall,\exists,\wedge,\vee,\rightarrow,\neg$ in
place of Agda's logical connectives in this paper.
Agda provides an hierarchy of types annotated with universe labels
which we omit for simplicity.
Following Agda's terminology, we refer to an Agda type as a
\emph{set}, and reserve the term \emph{type} for \Btt's
types. We use $\Mprop$ as the type of sets that denote propositions;
$\nat$ for the \set of natural numbers; and $\bool$ for the
\set of Booleans $\Mtt$ and $\Mff$.
Induction-recursion is used to define the forcing interpretation
in~\Cref{sec:forcing}.
We do not discuss this further here and the interested reader is
referred to the Agda formalization of this forcing
interpretation~(\flink{forcing.lagda}) for further details.

\subsection{Worlds}
\label{sec:worlds}

To capture the time progression notion which underlines choice operators, \Btt is parameterized by a
Kripke frame~\cite{Kripke:1963a,Kripke:1965} defined as follows:
\begin{defi}[\llink{98}~Kripke Frame]
\label{def:worlds}
A Kripke frame consists of a \set of \emph{worlds} $\Mworld$ equipped
with a reflexive and transitive binary relation $\MextSYMB$.
\end{defi}

Let $\Mw$ range over $\Mworld$.
We sometimes write $\Mextp{\Mw'}{\Mw}$ for $\Mext{\Mw}{\Mw'}$.
Let $\Mwpred{\Mw}$ be the collection of predicates on world extensions, i.e., functions in
$\Mallp{\Mextp{\Mw'}{\Mw}}{\Mprop}$.
Note that due to $\MextSYMB$'s transitivity, if
$P\in\Mwpred{\Mw}$ then for every $\Mextp{\Mw'}{\Mw}$ it naturally
extends to a predicate in $\Mwpred{\Mw'}$.
We further define the following notations for quantifiers.
$\Mallw{\Mw}{P}$ states that $P\in\Mwpred{\Mw}$ is true for all
extensions of $\Mw$, i.e., $\Mapp{P}{\Mw'}$ holds in all worlds
$\Mextp{\Mw'}{\Mw}$.
$\Mexw{\Mw}{P}$ states that $P\in\Mwpred{\Mw}$ is true at an extension
of $\Mw$, i.e., $\Mapp{P}{\Mw'}$ holds for some world
$\Mextp{\Mw'}{\Mw}$.
For readability, we sometime write $\Mallww{\Mw}{\Mw'}{P}$ instead of
$\Mallw{\Mw}{\Mlam{\Mw'}{P}}$ and $\Mexww{\Mw}{\Mw'}{P}$ instead of
$\Mexw{\Mw}{\Mlam{\Mw'}{P}}$.

\subsection{\texorpdfstring{\Btt}{BoxTT}'s Syntax and Operational Semantics}
\label{sec:syntax}
\label{sec:operational-semantics}
\label{sec:choice}
\label{sec:syn-sem-fresh}

\begin{figure*}[!t]
  \begin{small}
    \begin{center}
\begin{tabular}[t]{l}
$\begin{array}{lrl@{\hspace{0.15in}}l@{\hspace{0.4in}}ll@{\hspace{0.15in}}l@{\hspace{0.4in}}ll@{\hspace{0.15in}}l}
    \NUPRLval\in\NTval
            & ::=  & \NUPRLtval & \mbox{(type)}
            & \mid & \NUPRLlam{\NUPRL{x}}{\NUPRLterm} & \mbox{(lambda)}
            & \mid & \Naxiom & \mbox{(constant)}\\
	    & \mid & \metanat{n} & \mbox{(number)}
            & \mid & \NUPRLinl{\NUPRLterm} & \mbox{(left injection)}
            & \multicolumn{3}{l}{\hspace*{-0.05in}\hidebox{$\mid$\hspace{0.1in}$\Mcn$\hspace{0.1in}\mbox{(choice name)}}}\\
            & \mid & \NUPRLpair{\NUPRLterm_1}{\NUPRLterm_2} & \mbox{(pair)}
            & \mid & \NUPRLinr{\NUPRLterm} & \mbox{(right injection)}
\end{array}$\vspace{0.05in}
\\
$\begin{array}{lrl@{\hspace{0.12in}}l@{\hspace{0.18in}}ll@{\hspace{0.12in}}l@{\hspace{0.18in}}ll@{\hspace{0.12in}}l}
  \NUPRLtval\in\NTtype
            & ::=  & \Nproduct{x}{t_1}{t_2} & \mbox{(product)}
	    & \mid & \NUPRLset{x}{t_1}{t_2} & \mbox{(set)}
            & \mid & \NUPRLunion{t_1}{t_2} & \mbox{(disjoint union)}\\
            & \mid & \Nsum{x}{t_1}{t_2} & \mbox{(sum)}
           & \mid & \NUPRLequality{t}{t_1}{t_2} & \mbox{(equality)}
            & \multicolumn{3}{l}{\hspace*{-0.04in}\hidebox{$\mid$\hspace{0.1in}$\NnoreadSYMB$\hspace{0.2in}\mbox{(no read)}}} \\
            & \mid & \NUPRLuniverse{i} & \mbox{(universe)}
            & \mid & \NUPRLnat & \mbox{(numbers)}
            & \multicolumn{3}{l}{\hspace*{-0.04in}\hidebox{$\mid$\hspace{0.1in}$\NnowriteSYMB$\hspace{0.2in}\mbox{(no write)}}} \\
            & \mid & \Nisect{t_1}{t_2} & \mbox{(intersection)}
            & \mid & \Nqsquash{t} & \mbox{(sub-singleton)}
            & \multicolumn{3}{l}{\hspace*{-0.04in}\hidebox{$\mid$\hspace{0.1in}$\Npure$\hspace{0.2in}\mbox{(pure)}}} \\

\end{array}$\vspace{0.05in}
\\
$\begin{array}{lrl@{\hspace{0.11in}}l@{\hspace{0.22in}}ll@{\hspace{0.11in}}l@{\hspace{0.22in}}ll@{\hspace{0.11in}}l}
  \NUPRLterm\in\NTterm
             & ::=  & \NUPRL{x} & \mbox{(variable)}
             & \mid & \NUPRLval & \mbox{(value)}
             & \multicolumn{3}{l}{\hspace*{-0.05in}\hidebox{$\mid$\hspace{0.1in}$\Nread{\myfbox{$\NUPRLterm$}}$\hspace{0.1in}\mbox{(read)}}}
             \\
             & \mid & \NUPRLapp{\myfbox{$\NUPRLterm_1$}}{\NUPRLterm_2} & \mbox{(application)}
             & \mid & \NUPRLfix{\myfbox{$\NUPRLterm$}} & \mbox{(fixpoint)}
             & \multicolumn{3}{l}{\hspace*{-0.05in}\hidebox{$\mid$\hspace{0.1in}$\Nchoose{\myfbox{$t_1$}}{t_2}$\hspace{0.1in}\mbox{(choose)}}} \\
             & \mid & \Nlet{x}{\myfbox{$t_1$}}{t_2} & \mbox{(eager)}
             & \mid & \Nsucc{\myfbox{$t$}} & \mbox{(successor)}
             & \multicolumn{3}{l}{\hspace*{-0.05in}\hidebox{$\mid$\hspace{0.1in}$\Nfresh{x}{t}$\hspace{0.1in}\mbox{(fresh)}}} \\
             & \mid &
             \multicolumn{8}{l}{\Nnatrec{\myfbox{$t_1$}}{t_2}{t_3} \hspace{0.22in} \mbox{(number recursor)}}
             \\
             & \mid & \multicolumn{8}{l}{\NUPRLspread{\myfbox{$\NUPRLterm_1$}}{\NUPRL{x}}{\NUPRL{y}}{\NUPRLterm_2} \hspace{0.22in} \mbox{(pair destructor)}}
             \\
             & \mid & \multicolumn{8}{l}{\NUPRLdecide{\myfbox{$\NUPRLterm$}}{\NUPRL{x}}{\NUPRLterm_1}{\NUPRL{y}}{\NUPRLterm_2}
             \hspace{0.22in}\mbox{(injection destructor)}}
             \\

\end{array}$
\end{tabular}

             \hrulefill

  $\begin{array}{l@{\hspace{0.4in}}l}
    \begin{array}{l@{\hspace{0.1in}}l@{\hspace{0.1in}}l}
      \NUPRLapp
          {(\NUPRLlam{\NUPRL{x}}{\NUPRL{t}})}
          {\NUPRL{u}}
     &
     \NstepwwSYMB{\shidebox{$\Mw$}}{\shidebox{$\Mw$}}
&
\NUPRLsuba{\NUPRL{t}}{\NUPRL{x}}{\NUPRL{u}}
\\

\NUPRLfix{\NUPRLval}
&
\NstepwwSYMB{\shidebox{$\Mw$}}{\shidebox{$\Mw$}}
&
\NUPRLapp{\NUPRLval}{\NUPRLfix{\NUPRLval}}
\\

\Nlet{x}{v}{t_2}
&
\NstepwwSYMB{\shidebox{$\Mw$}}{\shidebox{$\Mw$}}
&
\NUPRLsuba{t_2}{x}{v}
\\

\Nsucc{\metanat{n}}
&
\NstepwwSYMB{\shidebox{$\Mw$}}{\shidebox{$\Mw$}}
&
\metanat{1+n}
\\

\Nnatrec
    {\metanat{0}}
    {t_1}
    {t_2}
&
\NstepwwSYMB{\shidebox{$\Mw$}}{\shidebox{$\Mw$}}
&
t_1
\\

\Nnatrec
    {\metanat{1+n}}
    {t_1}
    {t_2}
&
\NstepwwSYMB{\shidebox{$\Mw$}}{\shidebox{$\Mw$}}
&
\NUPRLappb{t_2}{\metanat{n}}{\Nnatrec{\metanat{n}}{t_1}{t_2}}
\end{array}\hspace{0.35in}

&

\hspace*{-0.7in}\begin{array}{l@{\hspace{0.05in}}l@{\hspace{0.05in}}l@{\hspace{0.2in}}l}

\multicolumn{3}{l}{\NUPRLspread
    {\NUPRLpair{\NUPRLterm_1}{\NUPRLterm_2}}
    {\NUPRL{x}}
    {\NUPRL{y}}
    {\NUPRL{t}}
\NstepwwSYMB{\shidebox{$\Mw$}}{\shidebox{$\Mw$}}
\NUPRLsubb{\NUPRL{t}}{\NUPRL{x}}{\NUPRLterm_1}{\NUPRL{y}}{\NUPRLterm_2}}

\\

\NUPRLdecide
    {\NUPRLinl{\NUPRLterm}}
    {\NUPRL{x}}
    {t_1}
    {\NUPRL{y}}
    {t_2}
&
\NstepwwSYMB{\shidebox{$\Mw$}}{\shidebox{$\Mw$}}
&
\NUPRLsuba{t_1}{\NUPRL{x}}{\NUPRLterm}
\\
\NUPRLdecide
    {\NUPRLinr{\NUPRLterm}}
    {\NUPRL{x}}
    {t_1}
    {\NUPRL{y}}
    {t_2}
&
\NstepwwSYMB{\shidebox{$\Mw$}}{\shidebox{$\Mw$}}
&
\NUPRLsuba{t_2}{\NUPRL{y}}{\NUPRLterm}
\\

\multicolumn{3}{c}{\hidebox{$
  {\Nread{\Mcn}}
  \NstepwwSYMB{\Mw}{\Mw}
  \Mctot{\Mgetchoice{\Mw}{\Mcn}}$}}
\\

\multicolumn{3}{c}{\hidebox{$
{\Nchoose{\Mcn}{t}}
\NstepwwSYMB{\Mw}{\Mchoose{\Mw}{\Mcn}{t}}
\Naxiom$}}
\\

\multicolumn{3}{c}{\hidebox{$
\Nfresh{x}{t}
\NstepwwSYMB{\Mw}{\MstartNewChoice{\Mw}{\Mrc}}
\NUPRLsuba{t}{x}{\MnewChoice{\Mw}}$}}

\end{array}

\end{array}$
\end{center}
  \end{small}
\caption{Core syntax (above) and small-step operational semantics (below)}%
\label{fig:syntax}
\label{fig:operational-semantics}
\end{figure*}

\Cref{fig:syntax} presents \Btt's syntax and call-by-name operational
semantics,
where the blue boxes highlight the time-related components,
and where $x$ belongs to a \set of variables $\NTvar$.
For simplicity, numbers are considered to be primitive.
The constant $\Naxiom$ is there for convenience, and is
used in place of a term, when the particular term used is irrelevant.
The term $\Nlet{x}{t_1}{t_2}$ is a call-by-value operator that allows
evaluating $t_1$ to a value before proceeding with $t_2$.
Terms are evaluated according to the operational
semantics presented in \Cref{fig:syntax}'s lower part, which is
further discussed below.
In what follows, we use all letters as metavariables for terms.
Let $\NUPRLsuba{t}{x}{u}$ stand for the capture-avoiding substitution
of all the free occurrences of $x$ in $t$ by $u$.
In what follows, we use the following definitions, where $x$ does not
occur free in $t_2$ or $t_3$:
$$\begin{array}{l}
\begin{array}{lll}
\Nite{t_1}{t_2}{t_3} & \DEF & \NUPRLdecide{t_1}{x}{t_2}{x}{t_3}
\end{array}\vspace{0.05in}\\
\begin{array}{lll}
\Nsq{t_1}{t_2} & \DEF & \Nlet{x}{t_1}{t_2}
\\
\Nbtrue & \DEF & \NUPRLinl{\Naxiom}
\\
\Nbfalse & \DEF & \NUPRLinr{\Naxiom}
\\
\Nneg{t} & \DEF & \Nite{t}{\Nbfalse}{\Nbtrue}
\end{array}\hspace{0.6in}
\begin{array}{lll}
\Niszero{t} & \DEF & \Nnatrec{t}{\Nbtrue}{\NUPRLlam{m}{\NUPRLlam{r}{\Nbfalse}}}
\\
\Npred{t} & \DEF & \Nnatrec{t}{\metanat{0}}{\NUPRLlam{m}{\NUPRLlam{r}{m}}}
\\
\Nsub{t_1}{t_2} & \DEF & \Nnatrec{t_2}{t_1}{\NUPRLlam{m}{\NUPRLlam{r}{\Npred{r}}}}
\\
\Nlt{t_1}{t_2} & \DEF & \Nneg{\Niszero{\Nsub{t_2}{t_1}}}
\end{array}
\end{array}$$

Types are syntactic forms that are given semantics
in~\Cref{sec:forcing} via a forcing interpretation.
The type system contains standard types such as dependent products of
the form $\Nproduct{x}{t_1}{t_2}$ and dependent sums of the form
$\Nsum{x}{t_1}{t_2}$.
It also includes \emph{subsingleton} types of the form $\Nqsquash{t}$,
which turns a type $t$ into a subsingleton type that equates all
elements of $t$; and \emph{intersection} types of the form
$\Nisect{t_1}{t_2}$, which is inhabited by the inhabitants of both
$t_1$ and $t_2$.
For convenience we introduce the following definitions, where $x$ does
not occur free in $t_2$:
$$\begin{array}{lll}
\NUPRLfalse & \DEF & \NUPRLequality{\NUPRLnat}{\metanat{0}}{\metanat{1}}
\\
\NUPRLtrue & \DEF & \NUPRLequality{\NUPRLnat}{\metanat{0}}{\metanat{0}}
\\
\Nbool & \DEF & \NUPRLunion{\NUPRLtrue}{\NUPRLtrue} 
\end{array}
\qquad
\begin{array}{lll}
\NUPRLfun{t_1}{t_2} & \DEF & \Nproduct{x}{t_1}{t_2}
\\
\NUPRLnot{T} & \DEF & \NUPRLfun{T}{\NUPRLfalse} 
\end{array}
\qquad
\begin{array}{lll}
\Nassertb{T} & \DEF & \NUPRLequality{\Nbool}{T}{\Nbtrue}
\\
\NUPRLsquash{T} & \DEF & \NUPRLset{x}{\NUPRLtrue}{T}
\end{array}$$

\Cref{fig:operational-semantics}'s lower part presents \Btt's
small-step operational semantics,
where $\Nstepwwa{t_1}{t_2}{\Mw_1}{\Mw_2}$ expresses that $t_1$ reduces
to $t_2$ in one step of computation starting from the world $\Mw_1$
and possibly updating it so that the resulting world is $\Mw_2$ at the
end of the computation step.
Most computations do not modify the current world except
$\Nchoose{t_1}{t_2}$.
We omit the congruence rules that allow computing within terms such
as: if $\Nstepwwa{t_1}{t_2}{\Mw_1}{\Mw_2}$ then
$\Nstepwwa{\NUPRLapppar{t_1}{u}}{\NUPRLapppar{t_2}{u}}{\Mw_1}{\Mw_2}$,
and the boxed terms of the form $\myfbox{$t$}$ in
\Cref{fig:operational-semantics} indicate the arguments that have to
be reduced before the outer operators can be reduced.
We denote by $\NUPRLcomputestoSYMB$ the reflexive transitive closure
of~$\NstepSYMB$, i.e.,~$\NUPRLcomputestopp{\Mw_1}{\Mw_2}{a}{b}$ states
that~$a$ computes to~$b$ in~$0$ or more steps, starting from the world
$\Mw_1$ and updating it so that the resulting world is $\Mw_2$ at the
end of the computation.
We write $\NUPRLcomputestop{\Mw}{a}{b}$ for
$\Mexww{\Mw}{\Mw'}{\NUPRLcomputestopp{\Mw}{\Mw'}{a}{b}}$.
We also write $\NUPRLcomputestoall{\Mw}{a}{b}$ if~$a$ computes to~$b$
in all extensions of $\Mw$, i.e., if
$\Mallww{\Mw}{\Mw'}{\NUPRLcomputestop{\Mw'}{a}{b}}$.

\Btt includes time-progressing notions that rely on worlds to record
choices and provides operators to manipulate the choices stored in a
world, which we now recall. Choices are referred to through their
names. A concrete example of such choices are reference cells in
programming languages, where a variable name pointing to a reference
cell is the name of the corresponding reference cell.
\label{def:choice-name}%
To this end, \Btt's
computation system is parameterized by a \set $\Mcname$ of
\emph{choice names}, that is equipped with a decidable equality, and
an operator that given a list of names, returns a name not in the
list.
This can be given by, e.g., nominal sets~\cite{pitts2013nominal}. In
what follows we let
$\Mcn$ range over $\Mcname$, and take $\Mcname$ to be $\nat$ for
simplicity.
\Btt is further parameterized over abstract operators and properties
recalled in
\Cref{def:choices,def:new-choices,def:compatibility,def:choosing},
which we show how to instantiate in \Cref{ex:ref}.
Definitions such as \Cref{def:choices} provide axiomatizations
of operators, and in addition informally indicate their intended use.
Choices are defined abstractly as follows:

\begin{defi}[\llink{101}~Choices] %
\label{def:choices}
\label{def:non-trivial-choices}
Let $\Mctype\subseteq\NTterm$ be a \set of \emph{choices},\footnote{To guarantee that
  $\Mctype\subseteq\NTterm$, one can for example extend the syntax to
  include a designated constructor for choices, or require a coercion
  $\Mfun{\Mctype}{\NTterm}$. We opted for the latter in our
  formalization.} and let
$\Mc$ range over $\Mctype$.
We say that a computation system contains $\langle
\Mcname,\Mctype\rangle$-choices if there exists
a partial function
$\MgetchoiceSYMB\in\Mfun{\Mworld}{\Mfun{\Mcname}{\Mctype}}$ (\llink{104}).
Given $\Mw\in\Mworld$ and $\Mcn\in\Mcname$,  the returned
choice, if it exists, is meant to be the %
last choice made for
$\Mcn$ according to $\Mw$.
$\Mctype$ is said to be \emph{\Pnontrivial} if it
contains two values $\Mca$ and $\Mcb$, which are syntactically
different (\llink{107}).
\end{defi}

A choice name $\Mcn$ can be used in a computation to access (or
``read'') choices from a world as follows:
${\Nread{\Mcn}}
\NstepwwSYMB{\Mw}{\Mw}
\Mctot{\Mgetchoice{\Mw}{\Mcn}}$ (as shown in~\Cref{fig:operational-semantics}).
This allows getting the last $\Mcn$-choice from the current world $\Mw$.
Datatypes are by default inhabited by impure computations that can for
example read choices using $\Nread{t}$.
For example, the $\NUPRLnat$ type is the type of potentially impure
natural numbers that includes expressions of the form $\Nread{\Mcn}$,
when $\Mcn$'s choices are natural numbers, in addition to expressions
of the form $\metanat{0}$, $\metanat{1}$, etc.

Note that the above definition of $\MgetchoiceSYMB$ is a slight
simplification of the more general notion of choices presented
in~\cite{Cohen+Rahli:fscd:2022}. There, the $\MgetchoiceSYMB$ function
was of type $\Mfun{\Mworld}{\Mfun{\Mcname}{\Mfun{\nat}{\Mctype}}}$.
The additional $\nat$ component enables a more general notion of
choice operators, encompassing both references and choice
sequences~\cite{Kleene+Vesley:1965,Atten+Dalen:2002,troelstra1985choice,troelstra1977choice,Kreisel+Troelstra:1970,Veldman:2001,moschovakis1993intuitionistic}, which stem
from Brouwer's intuitionistic logic, and can be seen as reference
cells that record the history of all the values ever stored in the
cells.
In references, which is the notion of choices we especially focus on in this
paper, one only maintains the latest update and so the $\nat$
component becomes moot. Thus, for simplicity of presentation, we elide
the $\nat$ component in this paper, but full details are available in
the Agda implementation.

The $\NnoreadSYMB$ type is inhabited by expressions that when
computing to a value in a world $\Mw_1$, also compute to that value in
all extensions $\Mw_2$ of $\Mw_1$.
Intuitively, this captures expressions that ``do not read'' in the
sense that they can only make limited use of the $\NreadSYMB$ operator
as the values they return should be the same irrespective of the world
they start computing against.
For example, given a choice name $\Mcn$, whose choices are natural
numbers, $\Nread{\Mcn}$ does not inhabit $\NnoreadSYMB$ as it could
return different values in successive extensions, while
$\Nlet{x}{\Nread{\Mcn}}{\metanat{0}}$ inhabits $\NnoreadSYMB$ as it
always compute to $\metanat{0}$, even though it reads $\Mcn$.
The $\NnoreadSYMB$ type can be turned into a $\Nnoreadmod{\_}$
modality as follows: let $\Nnoreadmod{T}\DEF\Nisect{T}{\NnoreadSYMB}$.
We then obtain that $\Nread{\Mcn}$ is not a member of
$\Nnoreadmod{\NUPRLnat}$, while it is a member of $\NUPRLnat$, and
$\Nlet{x}{\Nread{\Mcn}}{\metanat{0}}$ is a member of both $\NUPRLnat$
and $\Nnoreadmod{\NUPRLnat}$.
More generally, $\Nnoreadmod{T}$ is a subtype of $T$ in the sense
that the inhabitants of $\Nnoreadmod{T}$ also inhabit $T$.

While $\Nnoreadmod{T}$ restricts the effects that $T$'s inhabitants
can have, it can still be inhabited by effectful computations (we saw
that $\Nlet{x}{\Nread{\Mcn}}{\metanat{0}}$ inhabits
$\Nnoreadmod{\NUPRLnat}$). Therefore, to capture pure expressions,
\Btt provides the term $\Npure$, which is the type of ``pure'' terms,
i.e.~terms that do no contain choice names. This type can be turned
into a modality as follows: $\Npuremod{T}\DEF\Nisect{T}{\Npure}$.
Therefore, $\Npuremod{\NUPRLnat}$ is the type of pure natural numbers,
i.e., natural numbers that do not contain choice names. Hence,
$\Npuremod{\NUPRLnat}$ is a subtype of $\Nnoreadmod{\NUPRLnat}$. For
example, $\Nlet{x}{\Nread{\Mcn}}{\metanat{0}}$ inhabits
$\Nnoreadmod{\NUPRLnat}$, while it does not inhabit
$\Npuremod{\NUPRLnat}$.

\Btt also includes the notion of a \emph{restriction}, which
allows assuming that the choices made for a given choice name all satisfy a pre-defined constraint. Here again we simplify the concept for choices without history tracking. 
\begin{defi}[\llink{110}~Restrictions]
\label{def:restrictions}
A restriction $\Mr\in\Mres$ is a pair
$\mytuple{\META{res},\META{d}}$ consisting of a function
$\META{res}\in{\Mfun{\Mctype}{\Mprop}}$ and
a default choice $\META{d}\in\Mctype$  such that
$(\Mapp{\META{res}}{\META{d}})$ holds.
Given such a pair $\Mr$, we write $\Mdres{\Mr}$ for $\META{d}$.
\end{defi}

Intuitively, $\META{res}$ specifies a restriction on the choices that
can be made at any point in time and $\META{d}$ provides a default choice
that meets this restriction (e.g., for reference cells, this default choice is used to
initialize a cell).
For example, the restriction
$\mytuple{{\Mlam{\Mc}{\Mctot{\Mc}\in\nat}},0}$
requires choices to be
numbers and provides $0$ as a default value.
To reason about restrictions, we require the existence of
a ``compatibility'' predicate as follows.
\begin{defi}[\llink{113}~\PCompatibility]
\label{def:compatibility}
We say that $\Mctype$~is \emph{\Pcompatible} 
if there exists a predicate 
$\McompatibleSYMB\in\Mfun{\Mcname}{\Mfun{\Mworld}{\Mfun{\Mres}{\Mprop}}}$,
intended to guarantee that restrictions are satisfied, and which is
preserved by $\MextSYMB$:
$\Mallp{\Mparam{\Mcn}{\Mcname}\Mparam{\Mw_1,\Mw_2}{\Mworld}\Mparam{\Mr}{\Mres}}
{\Mfun{\Mext{\Mw_1}{\Mw_2}}{\Mfun{\Mcompatible{\Mcn}{\Mw_1}{\Mr}}{\Mcompatible{\Mcn}{\Mw_2}{\Mr}}}}$.
\end{defi}

\Btt further requires the ability to create new choice names as follows.
\begin{defi}[\llink{116}~\PExtendability]
\label{def:new-choices}
We say that $\Mctype$~is \emph{\Pextendable} if there exists a function
$\MnewChoiceSYMB\in\Mfun{\Mworld}{\Mcname}$ (\llink{119}),
where $\MnewChoice{\Mw}$ is
intended to return a new choice name not  present
in~$\Mw$,
and a function
$\MstartNewChoiceSYMB\in\Mfun{\Mworld}{\Mfun{\Mres}{\Mworld}}$ (\llink{122}),
where $\MstartNewChoice{\Mw}{\Mr}$ is intended to return an
extension of $\Mw$ with the new choice
name $\MnewChoice{\Mw}$ with restriction $\Mr$,
satisfying the following properties:
\begin{itemize}
\item Starting a new choice extends the current world:
  $\Mallp{\Mparam{\Mw}{\Mworld}\Mparam{\Mr}{\Mres}}
  {\Mext{\Mw}{\MstartNewChoice{\Mw}{\Mr}}}$
\item Initially, the only possible choice is the default value of the
  given restriction, i.e.:\\
  $\Mallp{\Mparam{\Mr}{\Mres}\Mparam{\Mw}{\Mworld}\Mparam{\Mc}{\Mctype}}
  {\Mfun{\Mgetchoice{\MstartNewChoice{\Mw}{\Mr}}{\MnewChoice{\Mw}}=\Mc}{\Mc=\Mdres{\Mr}}}$
\item A choice is initially compatible with its restriction:\\
  $\Mallp{\Mparam{\Mw}{\Mworld}\Mparam{\Mr}{\Mres}}
  {\Mcompatible{\MnewChoice{\Mw}}{\MstartNewChoice{\Mw}{\Mr}}{\Mr}}$
\end{itemize}
\end{defi}

\Btt provides a corresponding object computation that internalizes the
metatheoretical $\MnewChoiceSYMB$ function, namely $\Nfresh{x}{t}$.
Intuitively, it selects a ``fresh'' choice name $\Mcn$ and instantiate
the variable $x$ with $\Mcn$.
Formally, it computes as follows:
$\begin{array}{lll}
\Nfresh{x}{t}
&
\NstepwwSYMB{\Mw}{\MstartNewChoice{\Mw}{\Mrc}}
&
\NUPRLsuba{t}{x}{\MnewChoice{\Mw}}
\end{array}$
(as presented in \Cref{fig:operational-semantics}).
here $\Mrc$ is the restriction $\mytuple{\Mlam{c}{(c\in\nat)},0}$,
which constrains the choices to be numbers, with default value
$0$.
Other restrictions could be supported, for example by adding different
$\NfreshSYMB$ symbols to the language and by selecting during
computation the appropriate restriction based on the $\NfreshSYMB$
operator at hand. This is however left for future work as we
especially focus here on the choices presented in
\Cref{ex:ref}, i.e., natural numbers.

\begin{remark}[Freshness]
The fresh operator used in~\cite{Rahli+Bickford:cpp:2016} computes
$\Nfresh{x}{a}$ by reducing $a$ to $b$, and then returning
$\Nfresh{x}{b}$, thereby never generating new fresh names.
As opposed to that fresh operator, which was based on nominal sets,
the one introduced in this paper cannot put back the ``fresh''
constructor at each step of the small step derivation, otherwise a
multi-step computation would not be able to use a choice name to keep
track of the modulus of continuity of a function across multiple
computation steps by recording it in the current world.
One consequence of this is that this fresh operator cannot guarantee
that it generates a truly ``fresh'' name that does not occur anywhere
else (therefore, it does not satisfy the nominal axioms). For example
$\NUPRLapp{(\Nfresh{x}{x})}{\Mcn}$ might generate the name $\Mcn$
because it does not occur in the local expression
$\Nfresh{x}{x}$.%
\end{remark}

Lastly, \Btt requires the ability to update a choice as follows.

\begin{defi}[\llink{125}~\PMutability]
\label{def:choosing}
We say that $\Mctype$~is \emph{\Pmutable} if there exists a function
$\MchooseSYMB\in\Mfun{\Mworld}{\Mfun{\Mcname}{\Mfun{\Mctype}{\Mworld}}}$
such that if $\Mw\in\Mworld$, $\Mcn\in\Mcname$, $\Mc\in\Mctype$, then
$\Mext{\Mw}{\Mchoose{\Mw}{\Mcn}{\Mc}}$.
\end{defi}

\Btt provides a corresponding object computation that internalizes the
metatheoretical $\MchooseSYMB$ function, namely $\Nchoose{\Mcn}{t}$.
Choosing a $\Mcn$-choice $t$ using
$\Nchoose{\Mcn}{t}$ results in a corresponding update of the current
world, namely $\Mchoose{\Mw}{\Mcn}{t}$. The computation returns
$\Naxiom$, which is reminiscent of reference updates in OCaml for
example, which are of type \texttt{unit}.
As mentioned in \Cref{def:choices}, we require
$\Mctype\subseteq\NTterm$ so that choices can be included in
computations.
In addition, because
$\MchooseSYMB\in\Mfun{\Mworld}{\Mfun{\Mcname}{\Mfun{\Mctype}{\Mworld}}}$,
for $\Mchoose{\Mw}{\Mcn}{t}$ to be well-defined for $t\in\NTterm$, we
require a coercion from $\NTterm$ to $\Mctype$ so that $t$ can be
turned into a choice, and $\MchooseSYMB$ can be applied to that
choice. This coercion is left implicit for readability. We further
require that applying this coercion to a choice $\Mc$ returns $\Mc$,
which is used to validate the assumption $\assc$ discussed in
\Cref{sec:assumptions}.

The type $\NnowriteSYMB$ is inhabited by expressions that when
computing to a value from a world $\Mw_1$ to a world $\Mw_2$, satisfy
$\Mw_2=\Mw_1$. Intuitively, this captures expressions that ``do not
write'' in the sense that they can only make limited use of the
$\NchooseSYMB$ operator as the resulting world at the end of the
computation should be the world the computation started in.
For example, given a choice name $\Mcn$, whose choices are natural
numbers, $\Nchoose{\Mcn}{(\Nread{\Mcn}+\metanat{1})}$ does not
inhabit $\NnowriteSYMB$ as the computation ends in a world different
from the initial world, where $\Mcn$ is incremented by one,
while
$\Nlet{x}{\Nread{\Mcn}}{(\Nsq{(\Nchoose{\Mcn}{(x+1)})}{(\Nchoose{\Mcn}{x})})}$
inhabits $\NnowriteSYMB$, even though it uses $\NchooseSYMB$, as
it always ends in the same world as the initial world since $\Mcn$ is
reset to its initial value.
The $\NnowriteSYMB$ type can be turned into a $\Nnowritemod{\_}$
modality as follows: let
$\Nnowritemod{T}\DEF\Nisect{T}{\NnowriteSYMB}$.
We also write $\Nnoreadwritemod{T}$ for $\Nnowritemod{\Nnoreadmod{T}}$.
We then obtain that $\Nchoose{\Mcn}{(\Nread{\Mcn}+\metanat{1})}$ is
not a member of $\Nnowritemod{\NUPRLtrue}$, while it is a member of
$\NUPRLtrue$, and
$\Nlet{x}{\Nread{\Mcn}}{\Nsq{(\Nchoose{\Mcn}{(x+1)})}{\Nchoose{\Mcn}{x}}}$
is a member of both $\NUPRLtrue$ and $\Nnowritemod{\NUPRLtrue}$.
More generally, both $\Nnowritemod{T}$ and $\Nnoreadmod{T}$ are
subtypes of $T$, and $\Npuremod{T}$ is a subtype of both
$\Nnowritemod{T}$ and $\Nnoreadmod{T}$.

From this point on, we will only discuss choices $\Mctype$~that are
{\Pcompatible}, {\Pextendable} and {\Pmutable}.
While the abstract notion of choice operators has many concrete instances,
this paper focuses on one concrete instance, namely mutable references.

\begin{example}[\flink{worldInstanceRef.lagda}~References]
\label{ex:ref}
Reference cells, which are values that allow a program to indirectly
access a particular object, are choice operators since they can
point to different objects over their lifetime.
Formally, we define references to numbers, $\REF$,  as follows:
\begin{description}
\item[\Hlink{def:choices}{Non-trivial Choices}]
Let $\Mcname\DEF\nat$ and $\Mctype\DEF\nat$, which is
\Pnontrivial, e.g., take $\Mca\DEF\metanat{0}$ and
$\Mcb\DEF\metanat{1}$.

\item[\Hlink{def:worlds}{Worlds}]
Worlds are lists of cells, where a cell is a quadruple of (1)~a
choice name, (2)~a restriction, (3)~a choice, and (4)~a Boolean
indicating whether the cell is mutable.
$\MextSYMB$ is the reflexive transitive closure of two operations
that allow (i)~creating a new reference cell, and (ii)~updating an
existing reference cell.
We define $\Mgetchoice{\Mw}{\Mcn}$ so that it simply accesses the content of the $\Mcn$
cell in $\Mw$.

\item[\Hlink{def:compatibility}{Compatible}]
$\Mcompatible{\Mcn}{\Mw}{\Mr}$ states that a reference cell named
$\Mcn$ with restriction $\Mr$ was created in the world $\Mw$ (using
  an operation of type (i) described above), and that the current
  value of the cell satisfies $\Mr$.

\item[\Hlink{def:new-choices}{Extendable}]
$\MnewChoice{\Mw}$ returns a reference name not occurring in $\Mw$;
 and $\MstartNewChoice{\Mw}{\Mr}$ adds a new reference cell to
$\Mw$ with
 name $\MnewChoice{\Mw}$ and restriction~$\Mr$ (using an operation of type (i)  mentioned above).

 \item[\Hlink{def:choosing}{Mutable}]
$\Mchoose{\Mw}{\Mcn}{\Mc}$ updates the reference $\Mcn$ with the
   choice $\Mc$ if $\Mcn$ occurs in $\Mw$, and otherwise returns $\Mw$
   (using an operation of type (ii)  mentioned above).
\end{description}
A coercion from $\NTterm$ to $\Mctype$ can then turn $\metanat{n}$
into $n$ and all other terms to $0$, which satisfies the requirement
that choices are mapped to the same choices.
\end{example}

Formally,
$\NUPRLcomputestopp{\Mw_1}{\Mw_2}{a}{b}$ is the reflexive and
transitive closure of $\NstepSYMB$, i.e., it holds if $a$ in world
$\Mw_1$ computes to $b$ in world $\Mw_2$ in $0$ or more steps.
Thanks to the properties of $\MstartNewChoiceSYMB$ presented in
\Cref{def:new-choices}, and the properties of $\MchooseSYMB$
in \Cref{def:choosing},  computations respect
$\MextSYMB$:
\begin{lem}[\llink{128}~Computations respect $\MextSYMB$]
\label{lem:comps-sat-ext}
If $\NUPRLcomputestopp{\Mw_1}{\Mw_2}{a}{b}$ then
$\Mext{\Mw_1}{\Mw_2}$.
\end{lem}

\subsection{Forcing Interpretation}
\label{sec:forcing}

\begin{figure*}[!t]
  \begin{small}
    \begin{center}
\begin{description}[leftmargin=*]
\item[Numbers]~\\
\begin{itemize}
\item $\Meqtypes{\Mw}{\NUPRLnat}{\NUPRLnat}\iff\metatrue$
\item $\Mintype{\Mw}{t}{t'}{\NUPRLnat}
\iff
\Mbarmodw
{\Mw}
{\Mw'}
{\Mexp{\Mparam{n}{\nat}}
  {\NUPRLmetaand
    {\NUPRLcomputestop{\Mw'}{t}{\metanat{n}}}
    {\NUPRLcomputestop{\Mw'}{t'}{\metanat{n}}}
  }
}$
\end{itemize}

\vspace{0.1cm}
\item[Products]~\\
\begin{itemize}
\item $\begin{array}[t]{l}\Meqtypes{\Mw}{\Nproduct{x}{A_1}{B_1}}{\Nproduct{x}{A_2}{B_2}}
 \iff\Mfam{\Mw}{A_1}{A_2}{B_1}{B_2}\end{array}$
\item $\begin{array}{l}\Mintype{\Mw}{f}{g}{\Nproduct{x}{A}{B}}
\iff
\Mbarmodw{\Mw}
         {\Mw'}
         {\Mallp{\Mparam{a_1,a_2}{\NTterm}}
           {\Mfun{\Mintype{\Mw'}{a_1}{a_2}{A}}{\Mintype{\Mw'}{\NUPRLapppar{f}{a_1}}{\NUPRLapppar{g}{a_2}}{\NUPRLsuba{B}{x}{a_1}}}}}
\end{array}$
\end{itemize}

\vspace{0.1cm}
\item[Sums]~\\
\begin{itemize}
\item $\begin{array}[t]{l}\Meqtypes{\Mw}{\Nsum{x}{A_1}{B_1}}{\Nsum{x}{A_2}{B_2}}
\iff\Mfam{\Mw}{A_1}{A_2}{B_1}{B_2}\end{array}$
\item $\Mintype{\Mw}{p_1}{p_2}{\Nsum{x}{A}{B}}
\iff
\Mbarmodw{\Mw}
        {\Mw'}
        {\Mexp{\Mparam{a_1,a_2,b_1,b_2}{\NTterm}}
          {\Mandd{\Mintype{\Mw'}{a_1}{a_2}{A}}
            {\Mintype{\Mw'}{b_1}{b_2}{\NUPRLsuba{B}{x}{a_1}}}
            {\NUPRLcomputestop{\Mw'}{p_1}{\NUPRLpair{a_1}{b_1}}}
            {\NUPRLcomputestop{\Mw'}{p_2}{\NUPRLpair{a_2}{b_2}}}}}$
\end{itemize}

\vspace{0.1cm}
\item[Sets]~\\
\begin{itemize}
\item $\begin{array}[t]{l}\Meqtypes{\Mw}{\NUPRLset{x}{A_1}{B_1}}{\NUPRLset{x}{A_2}{B_2}}
\iff\Mfam{\Mw}{A_1}{A_2}{B_1}{B_2}\end{array}$
\item $\Mintype{\Mw}{a_1}{a_2}{\NUPRLset{x}{A}{B}}
\iff
\Mbarmodw{\Mw}
         {\Mw'}
         {\Mexp{\Mparam{b_1,b_2}{\NTterm}}{\Mand{\Mintype{\Mw'}{a_1}{a_2}{A}}{\Mintype{\Mw'}{b_1}{b_2}{\NUPRLsuba{B}{x}{a_1}}}}}$
\end{itemize}

\vspace{0.1cm}
\item[Disjoint unions]~\\
\begin{itemize}
\item $\begin{array}[t]{l}\Meqtypes{\Mw}{\NUPRLunion{A_1}{B_1}}{\NUPRLunion{A_2}{B_2}}
\iff
\Mand{\Meqtypes{\Mw}{A_1}{A_2}}{\Meqtypes{\Mw}{B_1}{B_2}}\end{array}$
\item
  $\Mintype{\Mw}{a_1}{a_2}{\NUPRLunion{A}{B}}
  \iff
  \Mbarmodw{\Mw}
           {\Mw'}
           {\Mexp{\Mparam{u,v}{\NTterm}}
             {\Mor{(\Mandc{\NUPRLcomputestop{\Mw'}{a_1}{\NUPRLinl{u}}}
                 {\NUPRLcomputestop{\Mw'}{a_2}{\NUPRLinl{v}}}
                 {\Mintype{\Mw'}{u}{v}{A}})}
               {(\Mandc{\NUPRLcomputestop{\Mw'}{a_1}{\NUPRLinr{u}}}
                 {\NUPRLcomputestop{\Mw'}{a_2}{\NUPRLinr{v}}}
                 {\Mintype{\Mw'}{u}{v}{B}})}}}$
\end{itemize}

\vspace{0.1cm}
\item[Equalities]~\\
\begin{itemize}
\item $\begin{array}[t]{l}\Meqtypes{\Mw}{\NUPRLequalityb{A}{a_1}{b_1}}{\NUPRLequalityb{B}{a_2}{b_2}}
\iff
\Mandc{\Meqtypes{\Mw}{A}{B}}
{\Mallww{\Mw}{\Mw'}{\Mintype{\Mw'}{a_1}{a_2}{A}}}
{\Mallww{\Mw}{\Mw'}{\Mintype{\Mw'}{b_1}{b_2}{B}}}\end{array}$
\item
  $\Mintype{\Mw}{a_1}{a_2}{\NUPRLequalityb{A}{a}{b}}
  \iff
  \Mbarmodw{\Mw}
         {\Mw'}
         {\Mintype{\Mw'}{a}{b}{A}}$
\end{itemize}

\vspace{0.1cm}
\item[Subsingletons]~\\
\begin{itemize}
\item $\Meqtypes{\Mw}{\Nqsquash{A}}{\Nqsquash{B}}
\iff
\Meqtypes{\Mw}{A}{B}$
\item $\Mintype{\Mw}{a}{b}{\Nqsquash{A}}
\iff
\Mbarmodw{\Mw}
         {\Mw'}
         {\Mand{\Mintype{\Mw'}{a}{a}{A}}{\Mintype{\Mw'}{b}{b}{A}}}$
\end{itemize}

\vspace{0.1cm}
\item[Binary intersections]~\\
\begin{itemize}
\item $\begin{array}[t]{l}\Meqtypes{\Mw}{\Nisect{A_1}{B_1}}{\Nisect{A_2}{B_2}}
\iff
\Mand{\Meqtypes{\Mw}{A_1}{A_2}}{\Meqtypes{\Mw}{B_1}{B_2}}\end{array}$
\item
  $\Mintype{\Mw}{a_1}{a_2}{\Nisect{A}{B}}
  \iff
  \Mbarmodw{\Mw}
           {\Mw'}
           {\Mand{\Mintype{\Mw'}{a_1}{a_2}{A}}
             {\Mintype{\Mw'}{a_1}{a_2}{B}}}$
\end{itemize}

\vspace{0.1cm}
\item[No-read types]~\\
\begin{itemize}
\item $\Meqtypes{\Mw}{\NnoreadSYMB}{\NnoreadSYMB}
\iff
\metatrue$
\item $\Mintype{\Mw}{a}{b}{\NnoreadSYMB}
\iff
\Mbarmodw
    {\Mw}{\Mw'}
    {\Mallp{\Mparam{v}{\NTval}}
      {\Mand
        {(\Mfun{\NUPRLcomputestop{\Mw'}{a}{v}}{\NUPRLcomputestoall{\Mw'}{a}{v}})}
        {(\Mfun{\NUPRLcomputestop{\Mw'}{b}{v}}{\NUPRLcomputestoall{\Mw'}{b}{v}})}}}$
\end{itemize}

\vspace{0.1cm}
\item[No-write types]~\\
\begin{itemize}
\item $\Meqtypes{\Mw}{\NnowriteSYMB}{\NnowriteSYMB}
\iff
\metatrue$
\item $\Mintype{\Mw}{a}{b}{\NnowriteSYMB}
\iff
\Mbarmodw
    {\Mw}{\Mw'}
    {\Mallp{\Mparam{v}{\NTval}}
      {\Mand
        {(\Mfun{\NUPRLcomputestop{\Mw'}{a}{v}}{\NUPRLcomputestopp{\Mw'}{\Mw'}{a}{v}})}
        {(\Mfun{\NUPRLcomputestop{\Mw'}{b}{v}}{\NUPRLcomputestopp{\Mw'}{\Mw'}{b}{v}})}}}$
\end{itemize}

\vspace{0.1cm}
\item[Purity]~\\
\begin{itemize}
\item $\Meqtypes{\Mw}{\Npure}{\Npure}
\iff
\metatrue$
\item $\Mintype{\Mw}{a_1}{a_2}{\Npure}
\iff
\Mand
{(\Mexp{\Mparam{t}{\NTterm}}{\Mand{\NUPRLcomputestoallb{\Mw}{a_1}{t}}{\Nnonames{t}}})}
{(\Mexp{\Mparam{t}{\NTterm}}{\Mand{\NUPRLcomputestoallb{\Mw}{a_2}{t}}{\Nnonames{t}}})}$
\end{itemize}

\vspace{0.1cm}
\item[Modality closure]~\\
\begin{itemize}
\item $\Meqtypes{\Mw}{T_1}{T_2}
\iff\Mbarmodw{\Mw}
             {\Mw'}
             {\Mexp
               {\Mparam{T'_1,T'_2}{\NTterm}}
               {\NUPRLmetaandc
                 {\NUPRLcomputestoall{\Mw'}{T_1}{T'_1}}
                 {\NUPRLcomputestoall{\Mw'}{T_2}{T'_2}}
                 {\Meqtypes{\Mw'}{T'_1}{T'_2}}}}$
\item $\Mintype{\Mw}{t_1}{t_2}{T}
\iff
\Mbarmodw{\Mw}
          {\Mw'}
            {\Mexp
              {\Mparam{T'}{\NTterm}}
              {\NUPRLmetaand
                {\NUPRLcomputestoall{\Mw'}{T}{T'}}
                {\Mintype{\Mw'}{t_1}{t_2}{T'}}}}$
\end{itemize}
\end{description}
\end{center}
\end{small}
\caption{Forcing Interpretation}
\label{fig:forcing}
\end{figure*}

\Btt's semantics, presented in \Cref{fig:forcing}, while similar to
the one presented
in~\cite{Cohen+Rahli:fscd:2022,Cohen+Rahli:csl:2023}, differs in one
major way, namely types are here impure by default.
It is interpreted via a forcing interpretation in which the forcing
conditions are worlds.
This interpretation is defined using induction-recursion as follows:
(1)~the inductive relation $\Meqtypes{\Mw}{T_1}{T_2}$
expresses type equality in the world $\Mw$; (2)~the recursive function
$\Mintype{\Mw}{t_1}{t_2}{T}$ expresses equality in a type.
We further use the following abstractions:
$$\begin{array}{l}
\begin{array}{lll}
\Mistype{\Mw}{T} & \DEF & \Meqtypes{\Mw}{T}{T}
\\
\Mmemtype{\Mw}{t}{T} & \DEF & \Mintype{\Mw}{t}{t}{T}
\end{array}\hspace{0.5in}
\begin{array}{lll}
\Minhtype{\Mw}{T} & \DEF & \Mexp{\Mparam{t}{\NTterm}}{\Mmemtype{\Mw}{t}{T}}
\\
\NUPRLcomputestoallb{\Mw}{a}{b} & \DEF & \Mallww{\Mw}{\Mw'}{\NUPRLcomputestopp{\Mw'}{\Mw'}{a}{b}}
\end{array}\vspace{0.05in}
\\
\begin{array}{l}
\Mfam{\Mw}{A_1}{A_2}{B_1}{B_2}
\DEF \\ \hspace*{0.1in}
\Mand{\Meqtypes{\Mw}{A_1}{A_2}}
{\Mallwwp{\Mw}{\Mw'}{\Mallp{\Mparam{a_1,a_2}{\NTterm}}
    {\Mfun{\Mintype{\Mw'}{a_1}{a_2}{A_1}}{\Meqtypes{\Mw'}{\NUPRLsuba{B_1}{x}{a_1}}{\NUPRLsuba{B_2}{x}{a_2}}}}}}
\end{array}
\end{array}$$
Note that $\NUPRLcomputestoallb{\Mw}{a}{b}$ captures the fact that the
computation does not change the initial world (this is used in
\Cref{thm:TS}).
\Cref{fig:forcing} defines in particular the semantics of $\Npure$,
which is inhabited by name-free terms, where $\Nnonames{t}$ is defined
recursively over~$t$ and returns false iff~$t$ contains a choice
name~$\Mcn$ or a fresh operator of the form $\Nfresh{x}{t}$.
This forcing interpretation is parameterized by a family of abstract
modalities $\MbarmodSYMB$, which we sometimes refer to simply as a
modality, which is a function that takes a world $\Mw$ to its modality
$\MbarmodM{\Mw}\in\Mfun{\Mwpred{\Mw}}{\Mprop}$.
We often write $\Mbarmodw{\Mw}{\Mw'}{P}$ for $\Mbarmod{\Mw}{\Mlam{\Mw'}{P}}$.
As in~\cite{Cohen+Rahli:fscd:2022}, to guarantee that this
interpretation yields a standard type system in the sense
of~\Cref{thm:TS}, we require that the modalities satisfy certain
properties reminiscent of standard modal axiom
schemata~\cite{Cresswell+Hughes:1996},
which we repeat here for ease of read:
\begin{defi}[\llink{133}~Equality modality]
\label{def:modalities}
The modality $\MbarmodSYMB$
is called an \emph{equality modality} if it satisfies the following
properties:
\begin{small}
\begin{itemize}
\label{def:mod-props}
\item $\modaa$ (monotonicity):
  $\Mallp{\Mparam{\Mw}{\Mworld}\Mparam{P}{\Mwpred{\Mw}}}
  {\Mallp{\Mextp{\Mw'}{\Mw}}
    {\Mfun
      {\Mbarmod{\Mw}{P}}
      {\Mbarmod{\Mw'}{\Mwplift{\Mw'}{P}}}}}$.
\item $\modab$ (distribution):
  $\Mallb{\Mw}{\Mworld}{P,Q}{\Mwpred{\Mw}}
    {\Mfun
      {\Mbarmodw{\Mw}{\Mw'}{\Mfun{\Mapp{P}{\Mw'}}{\Mapp{Q}{\Mw'}}}}
      {\Mfun{\Mbarmod{\Mw}{P}}{\Mbarmod{\Mw}{Q}}}}$
\item $\modad$ (density):
  $\Mallb{\Mw}{\Mworld}{P}{\Mwpred{\Mw}}
  {\Mfun
    {\Mbarmodw{\Mw}{\Mw'}{\Mbarmod{\Mw'}{\Mwplift{\Mw'}{P}}}}
    {\Mbarmod{\Mw}{P}}}$
\item $\modae$ (weakening):
  $\Mallb{\Mw}{\Mworld}{P}{\Mwpred{\Mw}}
  {\Mfun
    {\Mallw{\Mw}{P}}
    {\Mbarmod{\Mw}{P}}}$
\item $\modaf$ (reflexivity):
  $\Mallb{\Mw}{\Mworld}{P}{\Mprop}
  {\Mfun
    {\Mbarmodw{\Mw}{\Mw'}{P}}
    {P}}$
\end{itemize}
\end{small}
\end{defi}

As mentioned in~\cite{Cohen+Rahli:fscd:2022}, modalities can be
derived from covering relations $\McovSYMB$, where
$\Mcov{\Mw}{\Mopen}$ captures that $\Mopen\in\Mfun{\Mworld}{\Mprop}$
``covers'' the world $\Mw$.
Covering relations are required to satisfy suitable intersection,
union, top, non-emptiness, and subset
properties~\cite[Def.22]{Cohen+Rahli:fscd:2022} to be able to derive
modalities from them.
We present below three examples of coverings that satisfy these
properties, namely Kripke, Beth, and Open coverings, from which
modalities can be derived (the result presented in
\Cref{sec:continuity-proof} holds in particular for the resulting
modalities):
\begin{example}
\label{ex:coverings}
\label{def:chains}
The Kripke covering (\llink{136}) is defined as follows, i.e., $\Mcovn{K}{\Mw}{\Mopen}$
whenever $\Mopen$ contains all the extensions of $\Mw$:
$$\Mcovn{K}{\Mw}{\Mopen}\DEF\Mallww{\Mw}{\Mw'}{\Mapppar{\Mopen}{\Mw'}}$$

The Beth covering (\llink{139}) is defined as follows, i.e., $\Mcovn{B}{\Mw}{\Mopen}$
whenever $\Mopen$ contains at least one world from each ``branch''
(captured by $\MchainSYMB$ below) starting from $\Mw$:
$$\Mcovn{B}{\Mw}{\Mopen}\DEF\Mall{c}{\Mchain{\Mw}}{\Mexp{\Mparam{n}{\nat}}{\Mexp{\Mext{\Mw'}{(\Mapp{c}{n})}}{\Mapppar{\Mopen}{\Mw'}}}}$$
where $\Mchain{\Mw}$ is the \set of sequences of worlds
in $\Mfun{\nat}{\Mworld}$ such that $\Mch\in\Mchain{\Mw}$ iff
(1)~$\Mext{\Mw}{\Mapp{\Mch}{0}}$,
(2)~for all $i\in\nat$, $\Mext{\Mapp{\Mch}{i}}{\Mapp{\Mch}{(i+1)}}$;
and (3)~$c$ is \emph{progressing}, which is formally defined
in~\cite[Def.25]{Cohen+Rahli:fscd:2022}, and informally captures that
there exists two worlds $\Mext{\Mw_1}{\Mw_2}$ along $c$ where $\Mw_2$
is not $\Mw_1$, and this infinitely often.

The Open covering (\llink{142}) is defined as follows, i.e., $\Mcovn{O}{\Mw}{\Mopen}$
whenever $\Mopen$ contains a world that extends each world that
extends $\Mw$:
$$\Mcovn{O}{\Mw}{\Mopen}\DEF\Mallww{\Mw}{\Mw_1}{\Mexww{\Mw_1}{\Mw_2}{\Mapppar{\Mopen}{\Mw_2}}}$$
\end{example}

The Beth and Open coverings are in particular suitable
to capture aspects of choice sequences, which as mentioned above, can
be seen as reference cells that include the history of all the values
ever stored in the cells.
As discussed in~\cite{Cohen+Rahli:fscd:2022}, while the Beth covering
is incompatible with classical reasoning the Open covering allows
validating axioms such as the Law of Excluded Middle.

\begin{thm}
\label{thm:TS}
Given a computation system with choices $\Mctype$ and an equality modality
${\MbarmodSYMB}$, \Btt is a standard type system in the sense that its
forcing interpretation induced by $\MbarmodSYMB$ satisfy the following
properties (where free variables are universally quantified):
\BM$$\begin{footnotesize}\begin{array}{l@{\hspace{0.16in}}l@{\hspace{0.16in}}l}
  \mbox{transitivity:}
  &
  \Mfun
      {\Meqtypes{\Mw}{T_1}{T_2}}
      {\Mfun
        {\Meqtypes{\Mw}{T_2}{T_3}}
        {\Meqtypes{\Mw}{T_1}{T_3}}
      }
  &
  \Mfun
      {\Mintype{\Mw}{t_1}{t_2}{T}}
      {\Mfun
        {\Mintype{\Mw}{t_2}{t_3}{T}}
        {\Mintype{\Mw}{t_1}{t_3}{T}}
      }\\
  \mbox{symmetry:}
  &
  \Mfun
      {\Meqtypes{\Mw}{T_1}{T_2}}
      {\Meqtypes{\Mw}{T_2}{T_1}}
  &
  \Mfun
      {\Mintype{\Mw}{t_1}{t_2}{T}}
      {\Mintype{\Mw}{t_2}{t_1}{T}}
  \\
  \mbox{computation:}
  &
  \Mfun
      {\Meqtypes{\Mw}{T}{T}}
      {\Mfun
        {\NUPRLcomputestoallb{\Mw}{T}{T'}}
        {\Meqtypes{\Mw}{T}{T'}}
      }
  &
  \Mfun
      {\Mintype{\Mw}{t}{t}{T}}
      {\Mfun
        {\NUPRLcomputestoallb{\Mw}{t}{t'}}
        {\Mintype{\Mw}{t}{t'}{T}}
      }
  \\
  \mbox{monotonicity:}
  &
\Mfun
{\Meqtypes{\Mw}{T_1}{T_2}}
{\Mfun
  {\Mext{\Mw}{\Mw'}}
  {\Meqtypes{\Mw'}{T_1}{T_2}}
}
  &
\Mfun
{\Mintype{\Mw}{t_1}{t_2}{T}}
{\Mfun
  {\Mext{\Mw}{\Mw'}}
  {\Mintype{\Mw'}{t_1}{t_2}{T}}
}
\\
  \mbox{locality:}
  &
\Mfun
{\Mbarmodw{\Mw}{\Mw'}{\Meqtypes{\Mw'}{T_1}{T_2}}}
{\Meqtypes{\Mw}{T_1}{T_2}}
&
\Mfun
{\Mbarmodw{\Mw}{\Mw'}{\Mintype{\Mw'}{t_1}{t_2}{T}}}
{\Mintype{\Mw}{t_1}{t_2}{T}}
\\
  \mbox{consistency:}
&
\multicolumn{2}{l}{\neg\Mmemtype{\Mw}{t}{\NUPRLfalse}}
\end{array}\end{footnotesize}$$\EM
\end{thm}

\begin{proof}
The proof relies on the properties of the equality modality. For example:
$\modaa$ is used to prove monotonicity (\llink{145}) when $\Meqtypes{\Mw}{T_1}{T_2}$
is derived by closing under $\MbarmodM{\Mw}$;
$\modab$ and $\modae$ are used, e.g., to prove the symmetry
(\llink{150}) and transitivity (\llink{155}) of
$\Mintype{\Mw}{t}{t'}{\NUPRLnat}$;
$\modad$ is used to prove locality (\llink{161}); and
$\modaf$ is used to prove consistency (\llink{166}).
\end{proof}

As indicated in \Cref{thm:TS}, and as opposed to the counterpart of
the theorem in~\cite{Cohen+Rahli:fscd:2022},
$\Meqtypes{\Mw}{T_1}{T_2}$ and $\Mintype{\Mw}{t_1}{t_2}{T}$ are no
longer closed under all computations. For example, when
$T\DEF\NUPRLnat$, if $\NUPRLcomputestoall{\Mw}{t}{t'}$ and
$\Mintype{\Mw}{t}{\metanat{n}}{\NUPRLnat}$, does not necessarily give
us that $\Mintype{\Mw}{t'}{\metanat{n}}{\NUPRLnat}$. An example is:
$$t\DEF(\Nsq{\Nchoose{\Mcn}{\metanat{1}}}{\Niflt{\Nread{\Mcn}}{\metanat{1}}{\metanat{0}}{\metanat{1}}})$$
which reduces to
$t'\DEF(\Niflt{\Nread{\Mcn}}{\metanat{1}}{\metanat{0}}{\metanat{1}})$
and also to $\metanat{1}$ in all worlds, i.e.,
$\NUPRLcomputestoall{\Mw}{t}{t'}$ and
$\NUPRLcomputestoall{\Mw}{t}{\metanat{n}}$ for all $\Mw$, and
therefore $\Mintype{\Mw}{t}{\metanat{1}}{\NUPRLnat}$. However $t'$ does
not reduce to~$\metanat{1}$ in all worlds, and therefore
$\Mintype{\Mw}{t'}{\metanat{1}}{\NUPRLnat}$ does not hold, because
$\Mcn$ could be initialized differently in different worlds, for
example with~$\metanat{0}$, in which case $t'$ would reduce to
$\metanat{0}$.
This example can be adapted to show that $\Meqtypes{\Mw}{T_1}{T_2}$ is
also not closed under all computations:
$$T\DEF(\Nsq{\Nchoose{\Mcn}{\metanat{1}}}{\Niflt{\Nread{\Mcn}}{\metanat{1}}{\NUPRLtrue}{\NUPRLfalse}})$$
which reduces to
$T'\DEF(\Niflt{\Nread{\Mcn}}{\metanat{1}}{\NUPRLtrue}{\NUPRLfalse})$
and also to $\NUPRLfalse$ in all worlds, i.e.,
$\NUPRLcomputestoall{\Mw}{T}{T'}$ and
$\NUPRLcomputestoall{\Mw}{T}{\NUPRLfalse}$ for all $\Mw$, and
therefore $\Meqtypes{\Mw}{T}{\NUPRLvoid}$. However, $T'$ does not
reduce to $\NUPRLfalse$ in all worlds, because $\Mcn$ could be
initialized differently in different worlds, and therefore
$\Meqtypes{\Mw}{T'}{\NUPRLvoid}$ does not hold.

However, the following holds by transitivity of
$\NUPRLcomputestoSYMB$:
$$\begin{array}{l}
\Mfun
{\NUPRLcomputestoall{\Mw}{t'}{t}}
{\Mfun
  {\Mintype{\Mw}{t}{t}{\NUPRLnat}}
  {\Mintype{\Mw}{t}{t'}{\NUPRLnat}}
}
\\
\Mfun
{\NUPRLcomputestoall{\Mw}{T'}{T}}
{\Mfun
  {\Meqtypes{\Mw}{T}{T}}
  {\Meqtypes{\Mw}{T}{T'}}
}
\end{array}$$

To summarize, \Btt is closed under the following computations:
\begin{itemize}
\item \llink{169}
  $\Mfun
  {\NUPRLcomputestoall{\Mw}{U}{U'}}
  {\Mfun
    {\NUPRLcomputestoall{\Mw}{T}{T'}}
    {\Mfun
      {\Meqtypes{\Mw}{T'}{U'}}
      {\Meqtypes{\Mw}{T}{U}}
    }
  }$.
\item \llink{176}
  $\Mfun
  {\NUPRLcomputestoallb{\Mw}{U}{U'}}
  {\Mfun
    {\NUPRLcomputestoallb{\Mw}{T}{T'}}
    {\Mfun
      {\Meqtypes{\Mw}{T}{U}}
      {\Meqtypes{\Mw}{T'}{U'}}
    }
  }$,
  where the restriction to $\NUPRLcomputestoallbSYMB{\Mw}$ is due to
  the counterexample provided above
\item \llink{183}
  $\Mfun
  {\NUPRLcomputestoallb{\Mw}{a}{a'}}
  {\Mfun
    {\NUPRLcomputestoallb{\Mw}{b}{b'}}
    {\Mfun
      {\Mintype{\Mw}{a}{b}{A}}
      {\Mintype{\Mw}{a'}{b'}{A}}
    }
  }$,
  where the restriction to $\NUPRLcomputestoallbSYMB{\Mw}$ is due to
  the counterexample provided above.
  In particular, as indicated in \Cref{thm:TS}, this semantics is
  closed under $\beta$-reduction, as $\beta$-reduction does not modify
  the current world, i.e.,
  $\NUPRLcomputestoallb{\Mw}{\NUPRLapp{(\NUPRLlam{x}{t_1})}{t_2}}{\NUPRLsuba{t_1}{x}{t_2}}$,
  for all world $\Mw$.
\item \llink{190}
  $\Mfun
  {\NUPRLcomputestoallb{\Mw}{a}{a'}}
  {\Mfun
    {\NUPRLcomputestoallb{\Mw}{b}{b'}}
    {\Mfun
      {\Mintype{\Mw}{a'}{b'}{A}}
      {\Mintype{\Mw}{a}{b}{A}}
    }
  }$.
\end{itemize}

\subsection{Different Levels of Effects}
\label{sec:different-kinds-of-effects}

As mentioned above, \Btt provides types that allow capturing different
levels and kinds of effects, which we summarize here using natural
numbers as an example:
\begin{itemize}
\item \emph{Read \& write}: The type $\NUPRLnat$ according to
  \Cref{fig:forcing} is the type of potentially ``fully'' effectful
  numbers that can both read and write, as they can modify the world
  they compute against (e.g.,
  $\Nsq{(\Nchoose{\Mcn}{\metanat{1}})}{\Mcn}$), and in addition can
  compute to different values in successive worlds, and therefore
  return different values depending on the values they read (e.g.,
  $\Nread{\Mcn}$).
\item \emph{Write-only}: The type $\Nnoreadmod{\NUPRLnat}$ according to
  \Cref{fig:forcing} is the type of potentially effectful numbers that
  can write but effectively cannot read (or in a limited way) because
  they are constrained to compute to the same number in all extensions
  of the current world, which therefore limits their use of reading,
  as for example, reading a reference cells is likely to influence the
  outcome of the computation (e.g.,
  $\Nsq{(\Nchoose{\Mcn}{\metanat{1}})}{\metanat{0}}$, which writes but
  does not read).
  Note that here and below, the read and write constraints are only
  indicative as for example $\Nsq{\Nread{\Mcn}}{\metanat{0}}$ reads
  but still belongs to $\Nnoreadmod{\NUPRLnat}$, and is considered
  write-only on the basis that it is observationally equivalent to
  $\metanat{0}$.%
\item \emph{Read-only}: $\Nnowritemod{\NUPRLnat}$ is the type of potentially
  effectful numbers that can read but effectively cannot write as
  $\NUPRLnat$ is the type of numbers that can both read and write
  (e.g., $\Nread{\Mcn}$, which reads but does not write).
\item \emph{No read or write}: The type $\Nnoreadwritemod{\NUPRLnat}$ is the
  type of potentially effectful numbers that can only use effects in
  a limited way and effectively cannot read or write because the
  $\Nnowritemod{\_}$ operator requires computations to end in the same
  worlds they started with, effectively preventing writes, while
  $\Nnoreadmod{\_}$ requires computations to compute to the same value
  in successive worlds, effectively preventing reads (e.g.,
  $\metanat{0}$, but also $\Nsq{\Nread{\Mcn}}{\metanat{0}}$, which
  reads in a limited way that does not affect the computation, or
  $\Nsq{(\Nchoose{\Mcn}{\Nread{\Mcn}})}{\metanat{0}}$, which writes in
  a limited way that does not affect the computation).
\item \emph{Pure}: $\Npuremod{\NUPRLnat}$ is the type of pure natural
  numbers, that do not contain choices, which is a syntactic
  restriction (e.g., $\metanat{0}$ but not
  $\Nsq{\Nread{\Mcn}}{\metanat{0}}$ as it contains a choice name
  $\Mcn$ even though this name does not affect the computation).
\end{itemize}

We therefore obtain the following inclusions, where we write
$T\subseteq{U}$ for $T$ is a subtype of $U$, i.e., all equal members
of $T$ are equal members of $U$:
$\Npuremod{\NUPRLnat}\subseteq\Nnowritemod{\NUPRLnat}\subseteq\NUPRLnat$ and
$\Npuremod{\NUPRLnat}\subseteq\Nnoreadmod{\NUPRLnat}\subseteq\NUPRLnat$.

\subsection{Inference Rules}

In \Btt, sequents are of the form
$\Nsequentext
{H}
{T}
{t}$,
where $H$ is a list of hypotheses $h_1,\dots,h_n$, and a hypothesis
$h$ is of the form $\NUPRLhyp{x}{A}$, where the variable $x$ stands
for the name of the hypothesis and $A$ its type.
Such a sequent denotes that, assuming $h_1,\dots,h_n$, the term $t$ is
a member of the type $T$, and that therefore $T$ is a type. The
term~$t$ in this context is called the \emph{realizer}
of~$T$.
Realizers are sometimes omitted when irrelevant to the discussion, in
particular when $T$ is an equality type, which can be realized by any
term according to their semantics in \Cref{fig:forcing}.
A rule is a pair of a conclusion sequent $S$ and a list of premise
sequents, $S_{1},\cdots, S_{n}$ (written as usual using a fraction
notation, with the premises on top).
\Cref{sec:inf-rules-appx} provides a sample of \Btt's key inference
rules, while \Cref{sec:principles} provides examples of rules that
only hold for some instances of \Btt, and \Cref{sec:assumptions}
presents the continuity rule.

\section{Principles (In)Compatible with \texorpdfstring{\Btt}{BoxTT}}
\label{sec:principles}

\Btt enables the study of various effectful type theories and
their associated theories and in particular, their (in)compatibility
with classical reasoning.
In fact, in~\cite{Cohen+Rahli:fscd:2022} we have identified instances
of \Btt that are agnostic with respect to classical reasoning, in that
they are consistent with the Law of Excluded Middle (LEM)~(\flink{lem.lagda}),
and other instances that are incompatible with classical reasoning, in
that they allow validating the \emph{negation} of LEM~(\flink{not\_lem.lagda}),
the Limited Principle of Omniscience
(LPO)~\cite[p.9]{bishop1967foundations}~(\flink{not\_lpo.lagda}),
and Markov's Principle (MP)~(\flink{not\_mp.lagda}).%

In particular, we have shown that instantiating $\Mctype$ with either
references of choice sequences and $\MbarmodSYMB$ with a Beth modality,
derived from the Beth covering presented in \Cref{ex:coverings}, leads
to theories that invalidate LEM, i.e., such that the following holds:
$$\Mallp{\Mparam{\Mw}{\Mworld}}
{\neg\Minhtype{\Mw}{\Nproduct{P}{\NUPRLuniverse{i}}{\Nsquash{\NUPRLunion{P}{\NUPRLnot{P}}}}}}$$
The following inference rule is therefore consistent with those
instances of \Btt:
$$\infer[]
{\Nsequent{\NUPRL{H}}{\NUPRLnot{\Nproduct{P}{\NUPRLuniverse{i}}{\Nsquash{\NUPRLunion{P}{\NUPRLnot{P}}}}}}}
{}$$
However, instantiating $\Mctype$ with either references of choice
sequences and $\MbarmodSYMB$ with an Open modality, derived from the
Open covering presented in \Cref{ex:coverings}, leads to theories
that are compatible with LEM, i.e., such that the following holds
(using classical reasoning in the metatheory to validate this axiom):
$$\Mallp{\Mparam{\Mw}{\Mworld}}
{\Minhtype{\Mw}{\Nproduct{P}{\NUPRLuniverse{i}}{\Nsquash{\NUPRLunion{P}{\NUPRLnot{P}}}}}}$$
The following inference rule is therefore consistent with those
instances of \Btt:
$$\infer[]
{\Nsequent{\NUPRL{H}}{\Nproduct{P}{\NUPRLuniverse{i}}{\Nsquash{\NUPRLunion{P}{\NUPRLnot{P}}}}}}
{}$$

While applying the $\Nnowritemod{\_}$ and $\Nnoreadmod{\_}$ modalities
to $\NUPRLunionSYMB$ does not change the above results, this is not
the case about MP and LPO.
Here we demonstrate how the statements and proofs presented
in~\cite{Cohen+Rahli:fscd:2022} can be translated to the current
modified theory.
In particular, the statments are translated so that they use
$\NUPRLnatnrw$ and $\Nboolnrw$, where the $\Nnoreadmod{\_}$ modality
was built in the semantics of the theory used
in~\cite{Cohen+Rahli:fscd:2022}.
The use of the $\Nnowritemod{\_}$ modality is further discussed below
since there was no object computation to extend a world
in~\cite{Cohen+Rahli:fscd:2022}.
We demonstrate this with the translation of MP:
$$\Mallp{\Mparam{\Mw}{\Mworld}}
{\neg\Minhtype
  {\Mw}
  {\Nproduct
    {f}
    {\NUPRLfun{\NUPRLnatnrw}{\Nboolnrw}}
    {\NUPRLfun
      {\NUPRLnot{\NUPRLnot{(\Nsum{n}{\NUPRLnatnrw}{\Nassert{\NUPRLapp{f}{n}}})})}}
      {\NUPRLsquash{\Nsum{n}{\NUPRLnatnrw}{\Nassert{\NUPRLapp{f}{n}}}}}}}}$$

The proof of the validity of the above statement, i.e., the negation
of MP, relies crucially on the fact that \Btt's computation system
includes stateful computations that can evolve non-deterministically
over time. In particular, as discussed above, \Btt not only supports
choices in the form of reference cells, but also choice sequences.
As opposed to reference cells which hold a mutable choice, a choice
sequence is an ever growing sequence of immutable choices. To access
the $n$th choice, the $\MgetchoiceSYMB$ function needs to be of type
$\Mfun{\Mworld}{\Mfun{\Mcname}{\Mfun{\nat}{\Mctype}}}$; the
$\NreadSYMB$ operator needs to be updated accordingly:
$\Nread{\NUPRLapppar{\Mcn}{n}}$; as well as its operational semantics
${\Nread{\NUPRLapppar{\Mcn}{n}}}\NstepwwSYMB{\Mw}{\Mw}\Mctot{\Mgetchoiceb{\Mw}{\Mcn}{n}}$.

This validity proof was possible in~\cite{Cohen+Rahli:fscd:2022}
because there was then no way for computations to extend the world
they compute in.
In that proof, $f$ is instantiated with
$\NUPRLlam{n}{\Nread{\NUPRLapppar{\Mcn}{n}}}$ for some choice sequence
name $\Mcn$, which requires choices to be Booleans. We can prove that
$\Mcn$ inhabits $\NUPRLfun{\NUPRLnatnrw}{\Nboolnrw}$, and in
particular we can use the no-read modality $\Nnoreadmod{\_}$, because
choices made by choice sequences are immutable, and so for a given
$n\in\NUPRLnatnrw$, the term $\Nread{\NUPRLapppar{\Mcn}{n}}$ will
always compute to the same Boolean.
We then prove that even though
$\NUPRLnot{\NUPRLnot{\Nsum{n}{\NUPRLnatnrw}{\Nassert{\NUPRLapp{f}{n}}}}}$
holds because in any world it is always possible to find an extension
where $\Mcn$ makes the $\Nbtrue$ choice,
$\NUPRLsquash{\Nsum{n}{\NUPRLnatnrw}{\Nassert{\NUPRLapp{f}{n}}}}$ does
not hold because there is no way to prove that $\Mcn$ will ever make
the choice $\Nbtrue$.
It then follows that the following inference rule is consistent with
\Btt where $\Mctype$ is instantiated with choices sequences, and
$\MbarmodSYMB$ with the Beth modality.
$$\infer[]
{\Nsequent{\NUPRL{H}}{\NUPRLnot{\Nproduct
    {f}
    {\NUPRLfun{\NUPRLnatnrw}{\Nboolnrw}}
    {\NUPRLfun
      {\NUPRLnot{\NUPRLnot{(\Nsum{n}{\NUPRLnatnrw}{\Nassert{\NUPRLapp{f}{n}}})})}}
      {\NUPRLsquash{\Nsum{n}{\NUPRLnatnrw}{\Nassert{\NUPRLapp{f}{n}}}}}}}}}
{}$$

Let us go back to the use of the $\Nnowritemod{\_}$ modality.
Because computations can now update the world they compute in, it
could be that a $n\in\NUPRLnat$ computes to a $\metanat{k}$ while
updating the world by making the choice $\Nbtrue$.
Therefore, while we can prove the validity of the negation of MP as
formulated above, we cannot validate the following version using the
same method:
$$\Mallp{\Mparam{\Mw}{\Mworld}}
{\neg\Minhtype{\Mw}{\Nproduct{f}{\NUPRLfun{\NUPRLnatnr}{\Nboolnr}}
    {\NUPRLfun{\NUPRLnot{\NUPRLnot{(\Nsum{n}{\NUPRLnatnr}{\Nassert{\NUPRLapp{f}{n}}})})}}
    {\NUPRLsquash{\Nsum{n}{\NUPRLnatnr}{\Nassert{\NUPRLapp{f}{n}}}}}}}}$$
Instantiating $f$ with $\NUPRLlam{n}{\Nread{\NUPRLapppar{\Mcn}{n}}}$
again, we can still prove that
$\NUPRLnot{\NUPRLnot{\Nsum{n}{\NUPRLnatnr}{\Nassert{\NUPRLapp{f}{n}}}}}$
holds because once again we can exhibit a world where the choice
$\Nbtrue$ is made.
However, we cannot prove anymore that
$\NUPRLsquash{\Nsum{n}{\NUPRLnatnr}{\Nassert{\NUPRLapp{f}{n}}}}$,
i.e.,
$\NUPRLsquash{\Nsum{n}{\NUPRLnatnr}{\Nassert{\Nread{\NUPRLapppar{\Mcn}{n}}}}}$,
does not hold because we can instantiate $n$ with
$\Nsq{(\Nchoose{\Mcn}{\Nbtrue})}{\metanat{x}}$, where the choice made
with $\Nchoose{\Mcn}{\Nbtrue}$ happens to be the $x$th choice.
Using the $\Nnowritemod{\_}$ modality prevents the use of
$\Nchoose{\Mcn}{\Nbtrue}$.

\section{Proof of Continuity}
\label{sec:continuity-proof}

Having covered \Btt, its support for effects, and its semantics in
\Cref{sec:background}, we now state the version of Brouwer's
continuity principle that we validate in this paper, along with its
effectful realizer. For this we first introduce the following
notation:
$\Npureproduct{a}{A}{B}\DEF\Nproduct{a}{\Npuremod{A}}{B}$,
which quantifies over pure elements of type $A$.
In addition, in the rest of this paper we write $\NUPRLnatr$ for
$\NUPRLnatnr$ for readability, and $\Nbaire$ and $\Nbaire_n$ stand for
$\NUPRLfun{\NUPRLnatr}{\NUPRLnatr}$ and
$\NUPRLfun{\NUPRLset{x}{\NUPRLnatr}{x<n}}{\NUPRLnatr}$, respectively.

\begin{thm}[\llink{197}~Continuity Principle]
\label{def:continuity+realizer}
The following continuity principle, referred to as $\contp$, is valid w.r.t.\ the semantics
presented in \Cref{sec:forcing} (further assuming the properties $\assa$, $\assb$,
and $\assc$, presented in \Cref{sec:assumptions}):
\begin{equation}
\label{equ:continuity-with-purity}
\Npureproduct
{F}
{\NUPRLfun{\Nbaire}{\NUPRLnatr}}
{\Npureproduct
  {\Mba}
  {\Nbaire}
  {\Nqsquash
    {\Nsum
      {n}
      {\NUPRLnatr}
      {\Npureproduct
        {\Mbb}
        {\Nbaire}
        {\Nimplies
          {(\NUPRLequality{\Nbaire_n}{\Mba}{\Mbb})}
          {(\NUPRLequality{\NUPRLnatr}{\NUPRLapppar{F}{\Mba}}{\NUPRLapppar{F}{\Mbb}})}
        }
      }
    }
  }
}
\end{equation}
and is inhabited by
\begin{equation}
\label{equ:mod-comp}
\NUPRLlam{F}{\NUPRLlam{\Mba}{\NUPRLpair{\Nmod{F}{\Mba}}{\NUPRLlam{\Mbb}{\NUPRLlam{e}{\Naxiom}}}}}
\end{equation}
where $\Nmod{F}{\alpha}$ is the modulus of continuity of the function
$F\in\NUPRLfun{\Nbaire}{\NUPRLnatr}$ at  $\alpha\in\Nbaire$ and is
computed by the following expression:
$$\begin{array}{lll}
\Nmod{F}{\alpha} & \DEF & \Nfresh{x}
{(\Nsq
  {(\Nchoose{x}{\metanat{0}})}
  {\Nsq
    {\NUPRLapppar{F}{\Nupd{x}{\alpha}}}
    {\Nsucc{\Nread{x}}}
  })
}\\
\Nupd{\Mcn}{\alpha}
& \DEF & \NUPRLlam
    {x}
    {(\Nlet
      {y}
      {x}
      {(\Nsq
        {(\Niflt{\Nread{\Mcn}}{y}{\Nchoose{\Mcn}{y}}{\Naxiom})}
        {\NUPRLapppar{\alpha}{y}})})}\end{array}$$
More precisely, the following is true for any world $\Mw$:
$$\Mmemtype{\Mw}{\NUPRLlam{F}{\NUPRLlam{\Mba}{\NUPRLpair{\Nmod{F}{\Mba}}{\NUPRLlam{\Mbb}{\NUPRLlam{e}{\Naxiom}}}}}}{\contp}$$
\end{thm}
The rest of this section describes the proof of this theorem.
First, we intuitively explain how
$\Nmod{F}{\alpha}$ computes the modulus of continuity of a function
$F$ at a point $\alpha$. This is done using the following steps:
\begin{enumerate}
\item selecting, using $\NfreshSYMB$, a fresh choice name $\Mcn$ (the
  variable $x$ gets replaced with the freshly generated name $\Mcn$
  when computing $\NmodSYMB$), with the appropriate restriction (here
  a restriction that requires choices to be numbers as mentioned in
  \Cref{sec:syn-sem-fresh});
\item setting $\Mcn$ to $0$ using $\Nchoose{x}{\metanat{0}}$ (where
  $x$ is $\Mcn$ when this expression computes);
\item applying $F$ to a modified version of $\alpha$,
  namely $\Nupd{\Mcn}{\alpha}$, which computes as $\alpha$, except
  that in addition it increases $\Mcn$'s value every time $\alpha$ is
  applied to a number larger than the last chosen one;
\item returning the last chosen number using
  $\Nread{x}$ (again $x$ is $\Mcn$ when this
  expression computes), increased by one in order to return a number
  higher than any number $F$ applies $\alpha$ to.
\end{enumerate}

Let us illustrate how $\NmodSYMB$ computes through the following
example:
\begin{example}
Given the following expressions:
$$\begin{array}{l}
F\DEF\NUPRLlam{\Mba}{\NUPRLapppar{\Mba}{\NUPRLapppar{\Mba}{\metanat{2}}}}\\
\Mba\DEF\NUPRLlam{n}{\Nsucc{n}}
\end{array}$$
the expression $\Nmod{F}{\Mba}$ computes as follows.
First, it generates a fresh name $\Mcn$, which is used to record the
values to which $\Mba$ is applied to.
Then, $\Mcn$ is initialized with $\metanat{0}$.
It then computes $\NUPRLapppar{F}{\Nupd{\Mcn}{\Mba}}$, i.e.,
$t_1\DEF\NUPRLapppar{\Nupd{\Mcn}{\Mba}}{\NUPRLapppar{\Nupd{\Mcn}{\Mba}}{\metanat{2}}}$.
Due to $\NupdSYMB$'s definition, the argument
$t_2\DEF\NUPRLapppar{\Nupd{\Mcn}{\Mba}}{\metanat{2}}$ is evaluated,
which in turn results in the evaluation of $\metanat{2}$.
Since~$\metanat{2}$ is a value, $\Mcn$'s value, which is $\metanat{0}$,
is compared to $\metanat{2}$, and since $2>0$, $\Mcn$ is updated with
$\metanat{2}$. The computation of $t_2$ proceeds by applying $\Mba$ to
$\metanat{2}$, which results in $\metanat{3}$.
Going back to the computation of $t_1$, once the argument $t_2$ has
been evaluated to $\metanat{3}$, $\Mcn$'s value, which is now
$\metanat{2}$ is updated with $\metanat{3}$ since $3>2$. The
computation of $t_1$ then proceeds by applying $\Mba$ to
$\metanat{3}$, which is the result returned by
$\NUPRLapppar{F}{\Nupd{\Mcn}{\Mba}}$.
Finally, $\NmodSYMB$ reads $\Mcn$'s value, which is now
$\metanat{3}$ and returns~$\metanat{4}$, which is higher
than the two values $\Mba$ is applied to, namely $\metanat{2}$ and
$\metanat{3}$.
\end{example}

We divide the proof of the validity of the continuity principle, i.e.,
that it is inhabited by the expression presented in
\Cref{equ:mod-comp}, into the following three components, where
$F\in\NUPRLfun{\Nbaire}{\NUPRLnatr}$ and $\Mba\in\Nbaire$:
\begin{itemize}
\item Proving that the modulus is a number, i.e., $\Nmod{F}{\Mba}\in\NUPRLnatr$;
\item Proving that $\Nmod{F}{\Mba}$ returns the highest number that $\Mba$
  is applied to in the computation it performs;
\item Given $\Mbb\in\Nbaire$, proving that $\NUPRLapppar{F}{\Mba}$ and
  $\NUPRLapppar{F}{\Mbb}$ return the same number if $\Mba$ and $\Mbb$
  agree up to $\Nmod{F}{\Mba}$.
\end{itemize}

\subsection{Purity}
\label{sec:purity}

The proof presented below relies, among other things, on the fact that
$\NUPRLequality{\Nbaire}{\Mba}{\Nupd{\Mcn}{\Mba}}$ given a fresh
choice name $\Mcn$, where $\Nupd{\Mcn}{\Mba}$ is used to ``probe'' $F$
in $\Nmod{F}{\Mba}$.
This equality is however not true for
$\Mba\in\NUPRLfun{\NUPRLnat}{\NUPRLnat}$. Consider for example the
function
$\Mba\DEF\NUPRLlam{n}{\Niflt{\Nread{\gamma}}{\metanat{1}}{\metanat{0}}{\metanat{1}}}$,
where $\gamma$ is a choice name. To prove
$\NUPRLequality{\NUPRLfun{\NUPRLnat}{\NUPRLnat}}{\Mba}{\Nupd{\Mcn}{\Mba}}$
we have to prove that
$\NUPRLequality{\NUPRLnat}{\NUPRLapppar{\Mba}{n}}{\NUPRLapppar{\Nupd{\Mcn}{\Mba}}{n}}$
for all $n\in\NUPRLnat$.
Given $n\DEF\Nsq{(\Nchoose{\gamma}{\metanat{1}})}{\metanat{1}}$, the term
$\NUPRLapppar{\Mba}{n}$ computes to $\metanat{0}$ in a world where
$\gamma$ is $\metanat{0}$ because $n$ is not evaluated, while
$\NUPRLapppar{\Nupd{\Mcn}{\Mba}}{n}$ first evaluates $n$, setting
$\gamma$ to $\metanat{1}$, and therefore computes to~$\metanat{1}$ in
any world.
To avoid this, instead of using $\NUPRLnat$, we use $\NUPRLnatr$,
thereby limiting the effects expressions can have.

According to $\NUPRLnatr$'s semantics, which is as follows:
$$\Mintype{\Mw}{t}{t'}{\NUPRLnatr}
\iff
\Mbarmodw
{\Mw}
{\Mw'}
{\Mexp{\Mparam{n}{\nat}}
  {\NUPRLmetaand
    {\NUPRLcomputestoall{\Mw'}{t}{\metanat{n}}}
    {\NUPRLcomputestoall{\Mw'}{t'}{\metanat{n}}}
  }
}$$
to prove that $\Nmod{F}{\Mba}\in\NUPRLnatr$ w.r.t.\ a world $\Mw$, we
have to prove that it computes to the same number in all extensions of
$\Mw$. However, this will not be the case if $F$ or $\alpha$ have side
effects. For example, if $F$ is
$\NUPRLlam{f}{\Nsq{\NUPRLapppar{f}{!{\Mcn_0}}}{\metanat{0}}}$, for
some choice name~$\Mcn_0$, then it could happen that~$f$ gets applied
to $\metanat{0}$ in some world $\Mw_1$ if $!{\Mcn_0}$
returns~$\metanat{0}$, and to~$\metanat{1}$ in some world
$\Mextp{\Mw_2}{\Mw_1}$ if $!{\Mcn_0}$ returns~$\metanat{1}$. As
$\Nmod{F}{\Mba}$ returns the highest number that $F$ applies its
argument to, then $\Nmod{F}{\Mba}$ would in this instance return
different numbers in different extensions, and would therefore not
inhabit $\NUPRLnatr$.

Therefore, to validate a version of continuity which requires the
modulus of continuity to be ``time-invariant'' as in
\Cref{equ:continuity-with-purity}, one can require that both $F$ and
$\Mba$ are pure (i.e., name-free) terms.
Thanks to $\NpureproductSYMB$, we get to assume that both $F$ and
$\Mba$ are in $\Npure$ and therefore are name-free.
Note that it would not be enough to use the following pattern:
$\Nproduct{F}{\NUPRLfun{\Nbaire}{\NUPRLnatr}}{\NUPRLfun{(\NUPRLequality{\Npure}{F}{F})}{\dots}}$,
because then for the continuity principle to even be a type, we would
have to prove that $F$ is name-free to prove that
$\NUPRLequality{\Npure}{F}{F}$ is a type, only knowing that
$F\in\NUPRLfun{\Nbaire}{\NUPRLnatr}$, due to the semantics of equality
types, which is not true in general, and which is why we instead use
$\Nisect{\_}{\Npure}$.

Let us now discuss a potential solution to avoid such a purity
requirement, as well as some difficulties it involves, which we leave
investigating to future work.
One could try to validate instead the following version of the
continuity axiom, which mixes the uses of $\NUPRLnatr$ and
$\NUPRLnat$, where
$\Nqbaire{n}=\NUPRLfun{\NUPRLset{x}{\NUPRLnatr}{\Nqlt{x}{n}}}{\NUPRLnatr}$,
assuming the
existence of some type $\Nqlt{x}{n}$ that can relate an
$x\in\NUPRLnatr$ with an $n\in\NUPRLnat$:
$$\Nproduct
{F}
{\NUPRLfun{\Nbaire}{\NUPRLnatr}}
{\Nproduct
  {\alpha}
  {\Nbaire}
  {\Nqsquash
    {\Nsum
      {n}
      {\NUPRLnat}
      {\Nproduct
        {\beta}
        {\Nbaire}
        {\Nimplies
          {(\NUPRLequality{\Nqbaire{n}}{\alpha}{\beta})}
          {(\NUPRLequality{\NUPRLnatr}{\NUPRLapppar{F}{\alpha}}{\NUPRLapppar{F}{\beta}})}
        }
      }
    }
  }
}$$

A first difficulty with this is the type $\Nqlt{x}{n}$, which to prove
that it holds in some world~$\Mw$ would require proving that~$x$ is
equal to all possible values that~$n$ can take in extensions of~$\Mw$.
Another related difficulty is that it is at present unclear whether
this principle can be validated constructively. More precisely, proving its
validity would require:
\begin{enumerate}
\item Proving that $\Nmod{F}{\Mba}\in\NUPRLnat$, which is now straightforward.
\item Next, we have to prove that
$\Nproduct
{\beta}
{\Nbaire}
{\Nimplies
  {(\NUPRLequality{\Nqbaire{\Nmod{F}{\Mba}}}{\alpha}{\beta})}
  {(\NUPRLequality{\NUPRLnatr}{\NUPRLapppar{F}{\alpha}}{\NUPRLapppar{F}{\beta}})}
}$,
i.e., given $\beta\in\Nbaire$ such that
$\NUPRLequality{\Nqbaire{\Nmod{F}{\Mba}}}{\alpha}{\beta}$, we have to prove
$\NUPRLequality{\NUPRLnatr}{\NUPRLapppar{F}{\alpha}}{\NUPRLapppar{F}{\beta}}$.
The assumption $\NUPRLequality{\Nqbaire{\Nmod{F}{\Mba}}}{\alpha}{\beta}$ tells us
that given $k\in\NUPRLnatr$ such that $\Nqlt{k}{\Nmod{F}{\Mba}}$,
then $\NUPRLequality{\NUPRLnatr}{\NUPRLapppar{\alpha}{k}}{\NUPRLapppar{\beta}{k}}$.
As mentioned above, for $\Nqlt{k}{\Nmod{F}{\Mba}}$ to be true, it must
be that $k$ is less than $\Nmod{F}{\Mba}$ in all extensions of the
current world. However, without the purity constraint,
$\Nmod{F}{\Mba}$ can compute to different numbers in different
extensions.
\end{enumerate}

Going back to our goal
$\NUPRLequality{\NUPRLnatr}{\NUPRLapppar{F}{\alpha}}{\NUPRLapppar{F}{\beta}}$,
given the semantics of $\NUPRLnatr$ presented above, to prove this it
is enough to assume that $\NUPRLapppar{F}{\Nupd{\Mcn}{\alpha}}$
computes to some number $\metanat{m}$ in some world~$\Mw$, and to
prove that $\NUPRLapppar{F}{\beta}$ also computes to $\metanat{m}$ in
$\Mw$.
We can then inspect the computation
$\NUPRLcomputestopp{\Mw}{\Mw_1}{\NUPRLapppar{F}{\Nupd{\Mcn}{\alpha}}}{\metanat{k}}$,
where $\Mcn$ is the name generated by $\Nmod{F}{\Mba}$,
and show that it can be converted into a computation
$\NUPRLcomputestopp{\Mw}{\Mw_2}{\NUPRLapppar{F}{\Mbb}}{\metanat{k}}$, by
replacing $\NUPRLapppar{\alpha}{\metanat{i}}$ with
$\NUPRLapppar{\Mbb}{\metanat{i}}$, whenever we encounter such an
expression.
To do this, we need to know that $\NUPRLapppar{\alpha}{\metanat{i}}$
and $\NUPRLapppar{\beta}{\metanat{i}}$ compute to the same number
using $\NUPRLequality{\Nqbaire{\Nmod{F}{\Mba}}}{\alpha}{\beta}$. However, we only
know that $\metanat{i}$ is less than $\Nmod{F}{\Mba}$ in~$\Mw$, which is
not enough to use this assumption, as $\metanat{i}$ might be greater
than $\Nmod{F}{\Mba}$ in some other world $\Mw'$.
We can address this issue using classical logic to prove that there
exists a $\Mextp{\Mw'}{\Mw}$ such that for all $\Mextp{\Mw''}{\Mw}$,
the smallest number that $\alpha$ is applied to in the computation of
$\Nmod{F}{\Mba}$ w.r.t.~$\Mw'$ is less than the number that
$\Nmod{F}{\Mba}$ computes to w.r.t.~$\Mw''$.
In the argument sketched above we can then use $\Mw'$ instead of $\Mw$.

\subsection{Assumptions}
\label{sec:assumptions}

Before we prove that the continuity principle is inhabited, we will
summarize here the assumptions we will be making to prove this result,
where $\Mr$ is a restriction that requires choices to be numbers:
$$\begin{array}{l@{\hspace{0.1in}}l}
(\assa~\llink{202}) & \Mallb{\Mw}{\Mworld}{P}{\Mwpred{\Mw}}
{\Mfun
  {\Mbarmod{\Mw}{P}}
  {\Mexw{\Mw}{P}}
}
\\
(\assb~\llink{205}) & \Mallp
{\Mparam{\Mcn}{\Mcname}\Mparam{\Mw}{\Mworld}\Mparam{n}{\nat}}
{\Mfun
  {\Mcompatible{\Mcn}{\Mw}{\Mr}}
  {\Mallww{\Mchoose{\Mw}{\Mcn}{\metanat{n}}}{\Mw'}{\Mexp{\Mparam{k}{\nat}}{\Mctot{\Mgetchoice{\Mw'}{\Mcn}}=\metanat{k}}}}
}
\\
(\assc~\llink{208}) & \Mallp
{\Mparam{\Mcn}{\Mcname}\Mparam{\Mw}{\Mworld}\Mparam{k}{\nat}}
{\Mfun
  {\Mcompatible{\Mcn}{\Mw}{\Mr}}
  {\Mgetchoice{\Mchoose{\Mw}{\Mcn}{\metanat{k}}}{\Mcn}=\metanat{k}}
}
\end{array}$$

$\assa$ requires the modality $\MbarmodSYMB$ to be non-empty in the
sense that for $\Mbarmod{\Mw}{P}$ to be true, it has to be true for at
least one extension of $\Mw$. This is true about all non-empty
covering relations (i.e., that satisfy the property that if
$\Mcov{\Mw}{\Mopen}$ then $\Mexww{\Mw}{\Mw'}{\Mapppar{\Mopen}{\Mw'}}$~\llink{211}),
and therefore about the Kripke, Beth, and Open modalities which are
derived from such coverings~\cite[Sec.6.2]{Cohen+Rahli:fscd:2022}.
$\assb$ requires that the ``last'' choice of a
$r$-compatible choice name $\Mcn$ is indeed a number.
$\assc$ guarantees that retrieving a choice that was just
made will return that choice.

The last two assumptions are true about $\REF$, the formalization of
references to numbers presented in \Cref{ex:ref}
(\Btt instantiated with a Kriple modality and
references satisfies these properties~\flink{contInstanceKripkeRef.lagda}).
In addition both are true about another kind of stateful computations,
namely a variant of the formalization of choice sequences
presented in~\cite[Ex.5]{Cohen+Rahli:fscd:2022}, where new choices are
pre-pended as opposed to being appended
in~\cite{Cohen+Rahli:fscd:2022} (\Btt instantiated with a Kriple modality and choice
sequences satisfies these properties~\flink{contInstanceKripkeCS.lagda}).

It then follows from \Cref{def:continuity+realizer} that the following
inference rule is consistent with the instances of \Btt mentioned above:
$$\infer[]
{\Nsequentext
  {\NUPRL{H}}
  {\contp}
  {\NUPRLlam{F}{\NUPRLlam{\Mba}{\NUPRLpair{\Nmod{F}{\Mba}}{\NUPRLlam{\Mbb}{\NUPRLlam{e}{\Naxiom}}}}}}
}
{}$$

\subsection{The Modulus is a Number}
\label{sec:mod-is-a-num}

In this section we prove that $\Nmod{F}{\Mba}\in\NUPRLnatr$. More
precisely, we prove the following:
\begin{thm}[\llink{216}~The Modulus is a Number]
\label{thm:modulus-is-a-number}
If $\Nnonames{F}$, $\Nnonames{\alpha}$,
$\Mmemtype{\Mw}{F}{\Nfun{\Nbaire}{\NUPRLnatr}}$, and
$\Mmemtype{\Mw}{\Mba}{\Nbaire}$, for some world $\Mw$, then
\begin{equation}
\label{equ:concl-mod-nat-a}
\Mbarmodw
{\Mw}
{\Mw'}
{\Mexp{\Mparam{n}{\nat}}{\NUPRLcomputestoall{\Mw'}{\Nmod{F}{\Mba}}{\metanat{n}}}}
\end{equation}
\end{thm}

To prove the above,
we will make use of the fact that
$\Mmemtype{\Mw}{\Nupd{\Mcn}{\alpha}}{\Nbaire}$
and therefore also
$\Mmemtype{\Mw}{\NUPRLapppar{F}{\Nupd{\Mcn}{\alpha}}}{\NUPRLnatr}$, i.e., by
the semantics of $\NUPRLnatr$ presented in \Cref{sec:purity}, we have for some fresh name $\Mcn$:
\begin{equation}
\label{equ:concl-mod-nat-b}
\Mbarmodw
{\Mw}
{\Mw'}
{\Mexp{\Mparam{n}{\nat}}{\NUPRLcomputestoall{\Mw'}{\NUPRLapppar{F}{\Nupd{\Mcn}{\alpha}}}{\metanat{n}}}}.
\end{equation}
But for this we first need to start
computing $\Nmod{F}{\Mba}$ to generate a fresh name $\Mcn$ according to the
current world.
If that current world is some world $\Mextp{\Mw'}{\Mw}$ (obtained, for
example, using $\modae$ from \Cref{def:mod-props} on \Cref{equ:concl-mod-nat-a}), then we need to
be able to get that $\NUPRLapppar{F}{\Nupd{\Mcn}{\alpha}}$ computes to
a number w.r.t.\ $\Mw'$, which \Cref{equ:concl-mod-nat-b} might not
provide.
This is the reason for assumption $\assa$.

Going back to the proof of \Cref{equ:concl-mod-nat-a}, we use
$\modae$, and have to prove
$\Mexp{\Mparam{n}{\nat}}{\NUPRLcomputestoall{\Mw_1}{\Nmod{F}{\Mba}}{\metanat{n}}}$
for some $\Mextp{\Mw_1}{\Mw}$. We then:
\begin{enumerate}
\item[(A)]\label{proof-4.3-partA} first have to find a number $n$
such that $\Nmod{F}{\Mba}$ computes to $\metanat{n}$ w.r.t.\ $\Mw_1$,
\item[(B)]\label{proof-4.3-partB} and then show
that it does so also for all $\Mextp{\Mw'_1}{\Mw_1}$.
\end{enumerate}

Let us prove (A) first. We now start computing $\Nmod{F}{\Mba}$
w.r.t.\ $\Mw_1$. We generate a fresh name
$\Mcn\DEF\MnewChoice{\Mw_1}$, and have to prove that
$\Nsq
{(\Nchoose{\Mcn}{\metanat{0}})}
{\Nsq
  {\NUPRLapppar{F}{\Nupd{\Mcn}{\alpha}}}
  {\Nsucc{!{\Mcn}}}
}$
computes to a number w.r.t.\ $\Mw_2\DEF\MstartNewChoice{\Mw_1}{\Mr}$
that satisfies $\Mcompatible{\Mcn}{\Mw_2}{\Mr}$ (by the properties of
$\MstartNewChoiceSYMB$ presented in \Cref{def:new-choices}).
We keep computing this expression and have to prove that
$\Nsq
{\NUPRLapppar{F}{\Nupd{\Mcn}{\alpha}}}
{\Nsucc{!{\Mcn}}}$
computes to a number w.r.t.\ $\Mw_3\DEF\Mchoose{\Mw_2}{\Mcn}{\metanat{0}}$.

From $\assa$ and \Cref{equ:concl-mod-nat-b}, we
obtain $\Mextp{\Mw_5}{\Mw}$ and $n\in\nat$ such that
$\NUPRLcomputestoall{\Mw_5}{\NUPRLapppar{F}{\Nupd{\Mcn}{\alpha}}}{\metanat{n}}$,
from which we obtain by definition that there exists a $\Mw_6$ such
that
$\NUPRLcomputestopp{\Mw_5}{\Mw_6}{\NUPRLapppar{F}{\Nupd{\Mcn}{\alpha}}}{\metanat{n}}$.
Now, because $F$ and $\alpha$ are name-free, we can derive that there
exists a $\Mw_4$ such that
$\NUPRLcomputestopp{\Mw_3}{\Mw_4}{\NUPRLapppar{F}{\Nupd{\Mcn}{\alpha}}}{\metanat{n}}$
(\llink{225}).
It now remains to prove that
$\Nsq{\metanat{n}}{\Nsucc{!{\Mcn}}}$,
computes to a number w.r.t.\ $\Mw_4$.
It is then enough to prove that $!{\Mcn}$
computes to a number $\metanat{k}$ w.r.t.\ $\Mw_4$, in which case
$\Nsq{\metanat{n}}{\Nsucc{!{\Mcn}}}$
computes to $\metanat{k{+}1}$ w.r.t.~$\Mw_4$.
To prove this we make use of $\assb$ which states that $\Mr$ constrains the
$\Mcn$-choices to be numbers.
Using this and the facts that
$\Mcompatible{\Mcn}{\Mw_2}{\Mr}$ and $\Mext{\Mw_2}{\Mw_4}$ (by
$\MextSYMB$'s transitivity since $\Mext{\Mw_3}{\Mw_4}$ by
\Cref{lem:comps-sat-ext} and $\Mext{\Mw_2}{\Mw_3}$ by
\Cref{def:choosing}), we deduce that there exists a $k\in\nat$ such
that $\Mctot{\Mgetchoice{\Mw_4}{\Mcn}}=\metanat{k}$, and therefore
$!{\Mcn}$ computes to $\metanat{k}$
w.r.t.\ $\Mw_4$, and
$\Nsq{\metanat{n}}{\Nsucc{!{\Mcn}}}$
computes to $\metanat{k{+}1}$ w.r.t.~$\Mw_4$, which concludes the
proof of (A).

To prove
$\Mexp{\Mparam{n}{\nat}}{\NUPRLcomputestoall{\Mw_1}{\Nmod{F}{\Mba}}{\metanat{n}}}$,
we then instantiate the formula with $k{+}1$, and have to prove
$\NUPRLcomputestoall{\Mw_1}{\Nmod{F}{\Mba}}{\metanat{k{+}1}}$. We
already know that
$\NUPRLcomputestopp{\Mw_1}{\Mw_4}{\Nmod{F}{\Mba}}{\metanat{k{+}1}}$,
i.e., part (A) above, and we now prove part (B) above,
i.e., that it does so in all extensions of $\Mw_1$ too.

To prove (B) we assume a $\Mextp{\Mw'_1}{\Mw_1}$ and have to prove
that $\Nmod{F}{\Mba}$ computes to $\metanat{k{+}1}$ w.r.t.~$\Mw'_1$.
As before,  we start computing $\Nmod{F}{\Mba}$
w.r.t.\ $\Mw'_1$, and generate a fresh name
$\Mcn'\DEF\MnewChoice{\Mw'_1}$, and have to prove that
$\Nsq
{\NUPRLapppar{F}{\Nupd{\Mcn'}{\alpha}}}
{\Nsucc{!{\Mcn'}}}$
computes to $\metanat{k{+}1}$
w.r.t.\ $\Mw'_3\DEF\Mchoose{\Mw'_2}{\Mcn'}{\metanat{0}}$, where
$\Mw'_2\DEF\MstartNewChoice{\Mw'_1}{\Mr}$.
As $F$ and $\alpha$ are name-free,
$t_1\DEF\NUPRLapppar{F}{\Nupd{\Mcn}{\alpha}}$ and
$t_2\DEF\NUPRLapppar{F}{\Nupd{\Mcn'}{\alpha}}$ behave the same except
that when $t_1$ updates $\Mcn$ with a number, $t_2$ updates $\Mcn'$
with that number.

Using a syntactic simulation method, we will prove that because $t_1$
and $t_2$ are ``similar'' (which is captured by \Cref{def:sim1}
below), $\Mgetchoice{\Mw_3}{\Mcn}=\Mgetchoice{\Mw'_3}{\Mcn'}$,
and
$\NUPRLcomputestopp{\Mw_3}{\Mw_4}{t_1}{t'_1}$, then
$\NUPRLcomputestopp{\Mw'_3}{\Mw'_4}{t_2}{t'_2}$ such that $t'_1$ and
$t'_2$ are also ``similar'' and
$\Mgetchoice{\Mw_4}{\Mcn}=\Mgetchoice{\Mw'_4}{\Mcn'}$.
Note that $\Mgetchoice{\Mw_3}{\Mcn}$ and
$\Mgetchoice{\Mw'_3}{\Mcn'}$ return the same choice
because
$\Mgetchoice{\Mw_3}{\Mcn}=\Mgetchoice{\Mchoose{\Mw_2}{\Mcn}{\metanat{0}}}{\Mcn}=0$
and
$\Mgetchoice{\Mw'_3}{\Mcn'}=\Mgetchoice{\Mchoose{\Mw'_2}{\Mcn'}{\metanat{0}}}{\Mcn'}=0$.
To derive these equalities, we need assumption $\assc$ that relates
$\MgetchoiceSYMB$ and $\MchooseSYMB$.

Let us now define the simulation mentioned above:
\begin{defi}[\llink{237}]
\label{def:sim1}
The similarity relation $\Msimdiff{t_1}{t_2}{\Mcn_1}{\Mcn_2}{\alpha}$
is true  iff%
$$\begin{array}{ll}
     & (t_1=\Nupd{\Mcn_1}{\alpha}\wedge{t_2=\Nupd{\Mcn_2}{\alpha}})\\
\vee & (t_1=x\wedge{t_2=x}) \vee (t_1=\Naxiom\wedge{t_2=\Naxiom}) \vee (t_1=\metanat{n}\wedge{t_2=\metanat{n}})\\
\vee & (t_1=\NUPRLlam{x}{a}
       \wedge{t_2=\NUPRLlam{x}{b}}
       \wedge{\Msimdiff{a}{b}{\Mcn_1}{\Mcn_2}{\alpha}})\\
\vee & (t_1=(\NUPRLapp{a_1}{b_1})
       \wedge{t_2=(\NUPRLapp{a_2}{b_2})}
       \wedge{\Msimdiff{a_1}{a_2}{\Mcn_1}{\Mcn_2}{\alpha}}
       \wedge{\Msimdiff{b_1}{b_2}{\Mcn_1}{\Mcn_2}{\alpha}})\\
\vee & \dots
\end{array}$$
Most cases are omitted in this definition as they are similar to the ones
presented above. Note however that crucially terms of the form $\Mcn$
or $\Nfresh{x}{t}$ are not related, and that those are the only
expressions not related, thereby ruling out names except when occurring
inside $\NupdSYMB$ through the first clause.
\end{defi}
As discussed above, a key property of this similarity relation is as
follows, where all free variables are universally quantified, and
which we prove by induction on the computation
$\NUPRLcomputestopp{\Mw_1}{\Mw'_1}{t_1}{t'_1}$:
\begin{lem}[\llink{240}~$\MsimdiffSYMB$ is preserved by computations]
If $\Msimdiff{t_1}{t_2}{\Mcn_1}{\Mcn_2}{\alpha}$,
$\Nnonames{\alpha}$,\linebreak[4]
$\Mgetchoice{\Mw_1}{\Mcn_1}=\Mgetchoice{\Mw_2}{\Mcn_2}$,
$\NUPRLcomputestopp{\Mw_1}{\Mw'_1}{t_1}{t'_1}$,
$\Mcompatible{\Mcn_1}{\Mw_1}{\Mr}$, and
$\Mcompatible{\Mcn_2}{\Mw_2}{\Mr}$,
then there exist $\Mw'_2$ and $t'_2$ such that
$\NUPRLcomputestopp{\Mw_2}{\Mw'_2}{t_2}{t'_2}$,
$\Msimdiff{t'_1}{t'_2}{\Mcn_1}{\Mcn_2}{\alpha}$, and
$\Mgetchoice{\Mw'_1}{\Mcn_1}=\Mgetchoice{\Mw'_2}{\Mcn_2}$.
\end{lem}

We therefore obtain that there exist $t'_2$ and $\Mw'_4$ such that
$\NUPRLcomputestopp{\Mw'_3}{\Mw'_4}{\NUPRLapppar{F}{\Nupd{\Mcn'}{\alpha}}}{t'_2}$,
$\Msimdiff{\metanat{n}}{t'_2}{\Mcn}{\Mcn'}{\alpha}$ and
$\Mgetchoice{\Mw'_4}{\Mcn'}=\Mgetchoice{\Mw_4}{\Mcn}=\metanat{k}$.
Furthermore, by definition of the similarity relation, $t'_2=\metanat{n}$.
We obtain that
$\NUPRLcomputestopp
{\Mw'_3}
{\Mw'_4}
{\Nsq
  {\NUPRLapppar{F}{\Nupd{\Mcn'}{\alpha}}}
  {\Nsucc{!{\Mcn'}}}
}
{\Nsq
  {\metanat{n}}
  {\Nsucc{!{\Mcn'}}}}$
and so
$\NUPRLcomputestopp
{\Mw'_3}
{\Mw'_4}
{\Nsq
  {\NUPRLapppar{F}{\Nupd{\Mcn'}{\alpha}}}
  {\Nsucc{!{\Mcn'}}}
}
{\Nsucc{!{\Mcn'}}}$.
Because
$\Mgetchoice{\Mw'_4}{\Mcn'}=\metanat{k}$, we finally obtain
$\NUPRLcomputestopp
{\Mw'_3}
{\Mw'_4}
{\Nsq
  {\NUPRLapppar{F}{\Nupd{\Mcn'}{\alpha}}}
  {\Nsucc{!{\Mcn'}}}
}
{\metanat{k{+}1}}$,
which concludes the proof of (B), and therefore that $\Nmod{F}{\Mba}\in\NUPRLnatr$.

\subsection{The Modulus is the Highest Number}
\label{sec:mod-highest}

We now prove that $\Nmod{F}{\Mba}$ returns the highest number that
$\alpha$ is applied to in the computation it performs:
\begin{thm}[\llink{244}~The Modulus is the Highest Number]
If
$\NUPRLcomputestopp
{\Mw}
{\Mw'}
{\Nmod{F}{\Mba}}
{\metanat{n}}$
such that $\Nmod{F}{\Mba}$ generates a fresh name $\Mcn$ and
$\Mgetchoice{\Mw'}{\Mcn}=\metanat{i}$, then for any world $\Mw_0$
occurring along this computation, it must be that
$\Mgetchoice{\Mw_0}{\Mcn}=\metanat{j}$ such that $j\le{i}$.
\end{thm}
As shown above, we know that for any world $\Mw_1$ there exist
$\Mw_2\in\Mworld$ and $k\in\nat$ such that
$\NUPRLcomputestopp
{\Mw_1}
{\Mw_2}
{\Nmod{F}{\Mba}}
{\metanat{k{+}1}}$.
As in \Cref{sec:mod-is-a-num}, we start computing $\Nmod{F}{\Mba}$
w.r.t.\ the current world~$\Mw_1$, and generate a fresh name
$\Mcn\DEF\MnewChoice{\Mw_1}$, and deduce that
\begin{equation}
\label{equ:comp-mod-to-nat-a}
\NUPRLcomputestopp
{\Mw''_1}
{\Mw_2}
{\Nsq
  {\NUPRLapppar{F}{\Nupd{\Mcn}{\alpha}}}
  {\Nsucc{!{\Mcn}}}
}
{\metanat{k{+}1}}
\end{equation}
where $\Mw'_1\DEF\MstartNewChoice{\Mw_1}{\Mr}$ and
$\Mw''_1\DEF\Mchoose{\Mw'_1}{\Mcn}{\metanat{0}}$.
Furthermore, by $\assb$, there must be a
$n\in\nat$ such that $\Mgetchoice{\Mw_2}{\Mcn}=\metanat{n}$.

We now want to show that if $n<m$, for some $m\in\nat$ (which we will
instantiate with $k{+}1$), then it must also be that for any world
$\Mw$ along the computation in \Cref{equ:comp-mod-to-nat-a}, if
$\Mgetchoice{\Mw}{\Mcn}=\metanat{i}$ then $i<m$.
This is not true about any computation, but it is true about the above because
$\NupdSYMB$ only makes a choice if that choice is higher than the
``current'' one. To capture this, we define the property
$\Mupdterm{t}{\Mcn}{\alpha}$,
which captures that the only place where $\Mcn$ occurs in $t$ is
wrapped inside $\Nupd{\Mcn}{\alpha}$.
That is,
\label{def:updterm}
$\Mupdterm{t}{\Mcn}{\alpha}$ is true iff
$\Msimdiff{t}{t}{\Mcn}{\Mcn}{\alpha}$.
We can then prove the following result by induction on the computation
$\NUPRLcomputestopp{\Mw_1}{\Mw_2}{t}{u}$:
\begin{lem}[\llink{255}]
\label{lem:highest}
Let $\alpha$ be a closed name-free term, and $t$ be a term such that
$\Mupdterm{t}{\Mcn}{\alpha}$ and
$\NUPRLcomputestopp{\Mw_1}{\Mw_2}{t}{u}$, and let
$\Mgetchoice{\Mw_2}{\Mcn}=\metanat{n}$, such that $n<m$,
then for any world $\Mw$ along the computation
$\NUPRLcomputestopp{\Mw_1}{\Mw_2}{t}{u}$
if $\Mgetchoice{\Mw}{\Mcn}=\metanat{i}$ then $i<m$.
\end{lem}
Applying this result to
$\NUPRLcomputestopp
{\Mw''_1}
{\Mw_2}
{\Nsq
  {\NUPRLapppar{F}{\Nupd{\Mcn}{\alpha}}}
  {\Nsucc{!{\Mcn}}}
}
{\metanat{k{+}1}}$
and instantiating $m$ with $k{+}1$, we obtain that for any world $\Mw$
along that computation if $\Mgetchoice{\Mw}{\Mcn}=\metanat{i}$ then
$i<k{+}1$.

\subsection{The Modulus as a Modulus}

We now prove the crux of continuity, namely that $F$ returns the same
number on functions that agree up to $\Nmod{F}{\Mba}$:
\begin{thm}[\llink{261}~The Modulus as a Modulus]
If $\Mintype{\Mw}{\alpha}{\beta}{\Nbaire_n}$ then
$\Mintype{\Mw}{\NUPRLapppar{F}{\alpha}}{\NUPRLapppar{F}{\beta}}{\NUPRLnatr}$.
\end{thm}
First, we prove that
$\Mintype{\Mw}{\NUPRLapppar{F}{\alpha}}{\NUPRLapppar{F}{\Nupd{\Mcn}{\alpha}}}{\NUPRLnatr}$,
which follows from the semantics of $\NproductSYMB$ and $\NUPRLnatr$
presented in \Cref{fig:forcing}, and in particular the crucial fact that
$\Mintype{\Mw}{\alpha}{\Nupd{\Mcn}{\alpha}}{\Nbaire}$.
It is therefore enough to prove that
$\NUPRLapppar{F}{\Nupd{\Mcn}{\alpha}}$ and $\NUPRLapppar{F}{\beta}$
are equal in $\NUPRLnatr$.
Relating $\NUPRLapppar{F}{\Nupd{\Mcn}{\alpha}}$ and
$\NUPRLapppar{F}{\beta}$ instead of $\NUPRLapppar{F}{\alpha}$ and
$\NUPRLapppar{F}{\beta}$ allows getting access to the values
$\alpha$ gets applied to in the computation $\NUPRLapppar{F}{\alpha}$
as they are recorded using the choice name $\Mcn$.
We can then use these values to prove that
$\NUPRLapppar{F}{\Nupd{\Mcn}{\alpha}}$ and $\NUPRLapppar{F}{\beta}$
behave similarly up to applications of~$\alpha$ in the first
computation (i.e., the computation of
$\NUPRLapppar{F}{\Nupd{\Mcn}{\alpha}}$ to a value), which are
applications of~$\beta$ in the second (i.e., the computation of
$\NUPRLapppar{F}{\beta}$ to a value), and that these applications
reduce to the same numbers because the arguments, recorded using
$\Mcn$, are less than $\Nmod{F}{\Mba}$.

However, even though $\Nupd{\Mcn}{\alpha}$ and $\alpha$ are equal
in~$\Nbaire$, they behave slightly differently computationally as
$\Nupd{\Mcn}{\alpha}$ turns the call-by-name computations
$\NUPRLapppar{\alpha}{t}$ into call-by-value computations by first
reducing $t$ into an expression of the form $\metanat{i}$.
By typing, we know that $\NUPRLapppar{F}{\Nupd{\Mcn}{\alpha}}$ and
$\NUPRLapppar{F}{\beta}$ compute to numbers, and to relate the two
computations to prove that they compute to the same number, we first
apply a similar transformation to~$\NUPRLapppar{F}{\beta}$. Let
$\NforceSYMB$ be defined as follows:
\label{def:force-cbv}
$$\Nforce{f}\DEF\NUPRLlam{x}{\Nlet{y}{x}{\NUPRLapppar{f}{y}}}.$$
It is straightforward to derive that
$\Mintype{\Mw}{\NUPRLapppar{F}{\beta}}{\NUPRLapppar{F}{\Nforce{\beta}}}{\NUPRLnatr}$
from the semantics of $\NproductSYMB$ and $\NUPRLnatr$ presented in
\Cref{fig:forcing}, and in particular the crucial fact that
$\Mintype{\Mw}{\beta}{\Nforce{\beta}}{\Nbaire}$.
It is therefore enough to prove that
$\NUPRLapppar{F}{\Nupd{\Mcn}{\alpha}}$ and
$\NUPRLapppar{F}{\Nforce{\beta}}$ are equal in $\NUPRLnatr$.

Because
$\NUPRLcomputestopp
{\Mw}
{\Mw'}
{\NUPRLapppar{F}{\Nupd{\Mcn}{\alpha}}}
{\metanat{n}}$
(as explained in part~(A) in the proof of \Cref{thm:modulus-is-a-number}),
by \Cref{lem:highest} for any world $\Mw_0$ along this computation if
$\Mgetchoice{\Mw_0}{\Mcn}=\metanat{i}$ then $i<k{+}1$, where $k{+}1$ is
the number computed by $\Nmod{F}{\Mba}$.

We now prove that $\NUPRLapppar{F}{\Nupd{\Mcn}{\alpha}}$ and
$\NUPRLapppar{F}{\Nforce{\beta}}$ both compute to $\metanat{n}$ through another
simulation proof that relies on the following relation:
\begin{defi}[\llink{273}]
\label{def:sim2}
The similarity relation $\Msimforce{t_1}{t_2}{\Mcn}{\alpha}{\beta}$
is true iff%
$$\begin{array}{ll}
     & (t_1=\Nupd{\Mcn}{\alpha}\wedge{t_2=\Nforce{\beta}})\\
\vee & (t_1=x\wedge{t_2=x}) \vee (t_1=\Naxiom\wedge{t_2=\Naxiom}) \vee (t_1=\metanat{n}\wedge{t_2=\metanat{n}})\\
\vee & (t_1=\NUPRLlam{x}{a}
       \wedge{t_2=\NUPRLlam{x}{b}}
       \wedge{\Msimforce{a}{b}{\Mcn}{\alpha}{\beta}})\\
\vee & (t_1=(\NUPRLapp{a_1}{b_1})
       \wedge{t_2=(\NUPRLapp{a_2}{b_2})}
       \wedge{\Msimforce{a_1}{a_2}{\Mcn}{\alpha}{\beta}}
       \wedge{\Msimforce{b_1}{b_2}{\Mcn}{\alpha}{\beta}})\\
\vee & \dots
\end{array}$$
Most cases are omitted in this definition as they similar to the ones
presented above. Note however that crucially terms of the form $\Mcn$
or $\Nfresh{x}{t}$ are not related, and that those are the only
expressions not related, thereby ruling out names except when occurring
inside $\NupdSYMB$ through the first clause.
\end{defi}
A key property of this relation is as follows, which captures that
$\Msimforce{t_1}{t_2}{\Mcn}{\alpha}{\beta}$ is preserved by
computations, and which we prove by induction on the computation
$\NUPRLcomputestopp{\Mw}{\Mw'}{t_1}{t'_1}$:
\begin{lem}[\llink{276}]                 
\label{lem:sim-force-pres}
If
$\Msimforce{t_1}{t_2}{\Mcn}{\alpha}{\beta}$,
$\alpha$ and $\beta$ agree up to $k$,
$\NUPRLcomputestopp{\Mw}{\Mw'}{t_1}{t'_1}$ and for any world $\Mw_0$
along this computation if $\Mgetchoice{\Mw_0}{\Mcn}=\metanat{i}$
then $i<k{+}1$, then $\NUPRLcomputestopp{\Mw}{\Mw}{t_2}{t'_2}$ such
that $\Msimforce{t'_1}{t'_2}{\Mcn}{\alpha}{\beta}$.
\end{lem}
Therefore, because
$\Msimforce{\NUPRLapppar{F}{\Nupd{\Mcn}{\alpha}}}{\NUPRLapppar{F}{\Nforce{\beta}}}{\Mcn}{\alpha}{\beta}$
(as $F$ is name-free) and $\NUPRLapppar{F}{\Nupd{\Mcn}{\alpha}}$
computes to $\metanat{n}$, using \Cref{lem:sim-force-pres} it must be
that $\NUPRLapppar{F}{\Nforce{\beta}}$ also computes to $\metanat{n}$,
which concludes our proof of \Cref{def:continuity+realizer}, i.e., of
the validity of the continuity principle for pure functions.

\section{Conclusion and Related Works}
\label{sec:related-work}

The \Btt family of type theories is a uniform, flexible framework for studying effectful type theories. 
In this paper we focused on identifying a  subset of the \Btt family that allows for the internal validation of the continuity principle. 
That is, we have shown how, for this subset of theories, the modulus of
continuity of functions can be computed using an
expression of the underlying computation system.
This internalization required stateful computations, and we discuss some of the challenges arising from such impure computations.
As mentioned in the introduction, and as recalled below, this is not the first proof of
continuity, however to the best of our knowledge, this is the first
proof of an ``internal'' validity proof of continuity that relies on
stateful computations.
Furthermore, the proof presented above relies on an ``internal''
notion of probing through the use of stateful computations internal to
the computation language of the type theory, while approaches such as~
\cite{Coquand+Jaber:2010,Coquand+Jaber:2012,Xu+Escardo:2013} rely on a
meta-theoretic (or ``external'') notion of probing.

Troelstra proved in~\cite[p.158]{Troelstra:1973} that every closed
term $t\in\Nfun{\Nbaire}{\nat}$ of N-HA$^{\omega}$ has a provable
modulus of continuity in N-HA$^{\omega}$---see also~\cite{Beeson:1985}
for similar proofs of the consistency of continuity with various
constructive theories.

Coquand and Jaber~\cite{Coquand+Jaber:2010,Coquand+Jaber:2012} proved
the \emph{uniform} continuity of a Martin-L\"of-like intensional type
theory using forcing.
Their method consists in adding a generic element $\mathtt{f}$ as a
constant to the language that stands for a Cohen real of type
$2^{\nat}$, and defining the forcing conditions as approximations of
$\mathtt{f}$. They then define a suitable \emph{computability}
predicate that expresses when a term is a computable term of some type
up to approximations given by the forcing conditions.
The key steps are to (1) first prove that $\mathtt{f}$ is computable
and then (2) prove that well-typed terms are computable, from which
they derive uniform continuity: the uniform modulus of continuity is
given by the approximations.

Similarly, Escard\'o and Xu~\cite{Xu+Escardo:2013} proved that the
definable functionals of G\"odel's
system~T~\cite{Girad+Taylor+Lafont:1989} are uniformly continuous on
the Cantor space $\Ncantor\DEF\Mfun{\nat}{\bool}$ (without assuming classical logic or the Fan
Theorem).
For that, they developed the C-Space category, which internalizes
continuity, and has a \emph{Fan functional} which computes the modulus
of uniform continuity of functions in $\NUPRLfun{\Ncantor}{\nat}$.
Relating C-Space and the standard set-theoretical model of system~T,
they show that all System~T functions on the Cantor space are
uniformly continuous. Furthermore, using this model, they show how to
extract computational content from proofs in HA$^{\omega}$ extended
with a uniform continuity axiom, which is realized by the Fan
functional.

In~\cite{Escardo:2013}, Escard\'o proves that all System~T
functions are continuous on the Baire space and uniformly continuous
on the Cantor space without using forcing.
Instead, he provides an alternative interpretation of system~T, where
a number is interpreted by a \emph{dialogue tree}, which captures the
computation of a function w.r.t.\ an oracle of type $\Nbaire$.
Escard\'o first proves that such computations are continuous, and then
by defining a suitable relation between the standard interpretation
and the alternative one, that relates the interpretations of all
system~T terms, derives that for all system~T functions on the Baire
space are continuous.
While in~\cite{Escardo:2013}, dialogue trees are constructed and leave
outside of System~T, in~\cite{Escardo:internal-dialogue:2013},
Escard\'o showed that it is possible to derive System~T definable
Church-encoded dialogue trees, albeit still deriving these trees
outside of System~T (through a metatheoretic induction on System~T
terms).

Chuangjie later developed in~\cite{Chuangjie:lmcs:2020} a related
syntactic approach, that does not rely on dialogue trees, which allows
recovering the continuity of System-T functionals, as well as the fact
that moduli of continuity are T-definable.

In~\cite{Rahli+Bickford:cpp:2016,Rahli+Bickford:mscs:2017}, the
authors proved that Brouwer's continuity principle is consistent with
Nuprl~\cite{Constable+al:1986,Allen+al:2006} by realizing the modulus
of continuity of functions on the Baire space also using Longley's
method~\cite{Longley:1999}, but using exceptions instead of
references. The realizer there is more complicated than the one
presented in this paper as it involves an effectful computation that
repeatedly checks whether a given number is at least as high as the
modulus of continuity, and increasing that number until the modulus of
continuity is reached.
We do not require such a loop, and can directly extract the modulus of
continuity of a function.

In~\cite{Baillon+Mahboubi+Pedrot:csl:2022} the authors prove that all
BTT~\cite{Pedrot+Tabareau:lics:2017} functions are continuous by
generalizing the method used in~\cite{Escardo:2013}. Their model is
built in three steps as follows: an axiom model/translation adds an
oracle to the theory at hand; a branching model/translation interprets
types as intensional D-algebras, i.e., as types equipped with pythias;
and an algebraic parametricity model/translation that relates the two
previous translations by relating the calls to the pythia to the
oracle. Their models allows deriving that all functions are
continuous, but does not allow ``internalizing'' the continuity
principle, which is the goal of this paper.

Another version of the continuity principle,  the Inductive Continuity
Principle, has also been explored recently~\cite{Ghani+Hancock+Pattinson:cmcs:2006,Ghani+Hancock+Pattinson:lmcs:2009,Ghani+Hancock+Pattinson:mfps:2009,Escardo:2013}.
This principle is sometimes referred to as the Strong Continuity
Principle since it implies the continuity principle discussed in this
paper (and other variants). Roughly speaking, it states that for any
function $F$ from the Baire space to numbers, there exists a
(dialogue) tree that contains the values of $F$ at its leaves and such
that the modulus of $F$ at each point of the Baire space is given by
the length of the corresponding branch in the tree.
In~\cite{CRRT2023indCont} we have identified a subset of the \Btt
family that allows for the internal validation of the Inductive
Continuity Principle via computations that construct such dialogue
trees internally to the theories using effectful computations.
To prove finiteness of the computed trees and
termination of the overall program we there had to resort to (meta-)classical reasoning, and it remains to be seen if this can be avoided.

\bibliographystyle{alphaurl}
\bibliography{bibFCS,biblio,biblio2}

\ifx\tAppA\tAppB%

\newpage
\appendix

\section{Inference Rules}
\label{sec:inf-rules-appx}

The following provides a sample of \Btt's key inference rules.
In what follows, we write $\NUPRLmember{A}{a}$ for
$\NUPRLequality{A}{a}{a}$.
While similar rules had been formally verified
in~\cite{Anand+Rahli:2014} for a precursor of \Btt, the formal
verification of the validity of these rules w.r.t.\ \Btt's semantics
is left for future work.

\intitleb{Products}
The following rules are the standard $\NproductSYMB$-elimination rule,
$\NproductSYMB$-introduction rule, type equality for $\NproductSYMB$
types, and $\lambda$-introduction rule, respectively.
$$\begin{scriptsize}\begin{array}{l}\infer[]
{\Nsequentext{\NUPRL{H},\NUPRLhyp{f}{\Nproduct{x}{A}{B}},\NUPRL{J}}{C}{\NUPRLsuba{e}{z}{\Naxiom}}}
{
  \Nsequent{\NUPRL{H},\NUPRLhyp{f}{\Nproduct{x}{A}{B}},\NUPRL{J}}{\NUPRLmember{A}{a}}
  &
  \Nsequentext{\NUPRL{H},\NUPRLhyp{f}{\Nproduct{x}{A}{B}},\NUPRL{J},\NUPRLhyp{z}{\NUPRLmember{\NUPRLsuba{B}{x}{a}}{\NUPRLapppar{f}{a}}}}{C}{e}
}
\hspace{0.2in}
\infer[]
{\Nsequentext{\NUPRL{H}}{\Nproduct{x}{A}{B}}{\NUPRLlam{z}{b}}}
{
  \Nsequentext{\NUPRL{H},\NUPRLhyp{z}{A}}{\NUPRLsuba{B}{x}{z}}{b}
  &
  \Nsequent{\NUPRL{H}}{\NUPRLmember{\NUPRLuniverse{i}}{A}}
}\\ \\
\infer[]
{\Nsequent{\NUPRL{H}}{\NUPRLequality{\NUPRLuniverse{i}}{\Nproduct{x_1}{A_1}{B_1}}{\Nproduct{x_2}{A_2}{B_2}}}}
{
  \Nsequent{\NUPRL{H}}{\NUPRLequality{\NUPRLuniverse{i}}{A_1}{A_2}}
  &
  \Nsequent{\NUPRL{H},\NUPRLhyp{y}{A_1}}{\NUPRLequality{\NUPRLuniverse{i}}{\NUPRLsuba{B_1}{x_1}{y}}{\NUPRLsuba{B_2}{x_2}{y}}}
}\hspace{0.2in}
\infer[]
{\Nsequent{\NUPRL{H}}{\NUPRLequality{\Nproduct{x}{A}{B}}{\NUPRLlam{x_1}{t_1}}{\NUPRLlam{x_2}{t_2}}}}
{
  \Nsequent{\NUPRL{H},\NUPRLhyp{z}{A}}{\NUPRLequality{\NUPRLsuba{B}{x}{z}}{\NUPRLsuba{t_1}{x_1}{z}}{\NUPRLsuba{t_2}{x_2}{z}}}
  &
  \Nsequent{\NUPRL{H}}{\NUPRLmember{\NUPRLuniverse{i}}{A}}
}\end{array}\end{scriptsize}$$
Note that the last rule requires to prove that $A$ is a type because
the conclusion requires to prove that $\Nproduct{x}{A}{B}$ is a type,
and the first hypothesis only states that $B$ is a type family
over~$A$, but does not ensures that $A$ is a type.
Furthermore, the following rules are the standard function extensionality and
$\beta$-computation rules, respectively:
$$\begin{scriptsize}\begin{array}{l}\infer[]
{\Nsequent{\NUPRL{H}}{\NUPRLequality{\Nproduct{x}{A}{B}}{f_1}{f_2}}}
{
  \Nsequent{\NUPRL{H},\NUPRLhyp{z}{A}}{\NUPRLequality{\NUPRLsuba{B}{x}{z}}{\NUPRLapppar{f_1}{z}}{\NUPRLapppar{f_2}{z}}}
  &
  \Nsequent{\NUPRL{H}}{\NUPRLmember{\NUPRLuniverse{i}}{A}}
}
\hspace{0.2in}
\infer[]
{\Nsequent{\NUPRL{H}}{\NUPRLequality{T}{\NUPRLapp{(\NUPRLlam{x}{t})}{s}}{u}}}
{\Nsequent{\NUPRL{H}}{\NUPRLequality{T}{\NUPRLsuba{t}{x}{s}}{u}}}\end{array}\end{scriptsize}$$

\intitleb{Sums}
The following rules are the standard $\NsumSYMB$-elimination rule,
$\NsumSYMB$-introduction rule, type equality for the $\NsumSYMB$ type,
pair-introduction, and spread-computation rules, respectively:
$$\begin{scriptsize}\begin{array}{l}\infer[]
{\Nsequentext{\NUPRL{H},\NUPRLhyp{p}{\Nnoreadwritemod{\Nsum{x}{A}{B}}},\NUPRL{J}}{C}{\NUPRLspread{p}{a}{b}{e}}}
{\Nsequentext{\NUPRL{H},\NUPRLhyp{p}{\Nnoreadwritemod{\Nsum{x}{A}{B}}},\NUPRLhyp{a}{A},\NUPRLhyp{b}{\NUPRLsuba{B}{x}{a}},\NUPRLsuba{\NUPRL{J}}{p}{\NUPRLpair{a}{b}}}{\NUPRLsuba{C}{p}{\NUPRLpair{a}{b}}}{e}}
\hspace{0.2in}
\infer[]
{\Nsequentext{\NUPRL{H}}{\Nsum{x}{A}{B}}{\NUPRLpair{a}{b}}}
{
  \Nsequent{\NUPRL{H}}{\NUPRLmember{A}{a}}
  &
  \Nsequent{\NUPRL{H}}{\NUPRLmember{\NUPRLsuba{B}{x}{a}}{b}}
  &
  \Nsequent{\NUPRL{H},\NUPRLhyp{z}{A}}{\NUPRLmember{\NUPRLuniverse{i}}{\NUPRLsuba{B}{x}{z}}}
}\\ \\
\infer[]
{\Nsequent{\NUPRL{H}}{\NUPRLequality{\NUPRLuniverse{i}}{\Nsum{x_1}{A_1}{B_1}}{\Nsum{x_2}{A_2}{B_2}}}}
{
  \Nsequent{\NUPRL{H}}{\NUPRLequality{\NUPRLuniverse{i}}{A_1}{A_2}}
  &
  \Nsequent{\NUPRL{H},\NUPRLhyp{y}{A_1}}{\NUPRLequality{\NUPRLuniverse{i}}{\NUPRLsuba{B_1}{x_1}{y}}{\NUPRLsuba{B_2}{x_2}{y}}}
}
\hspace{0.1in}
\infer[]
{\Nsequent{\NUPRL{H}}{\NUPRLequality{\Nsum{x}{A}{B}}{\NUPRLpair{a_1}{b_1}}{\NUPRLpair{a_2}{b_2}}}}
{
  \Nsequent{\NUPRL{H},\NUPRLhyp{z}{A}}{\NUPRLmember{\NUPRLuniverse{i}}{\NUPRLsuba{B}{x}{z}}}
  &
  \Nsequent{\NUPRL{H}}{\NUPRLequality{A}{a_1}{a_2}}
  &
  \Nsequent{\NUPRL{H}}{\NUPRLequality{\NUPRLsuba{B}{x}{a_1}}{b_1}{b_2}}
}
\\ \\
\infer[]
{\Nsequent{\NUPRL{H}}{\NUPRLequality{T}{\NUPRLspread{\NUPRLpair{s}{t}}{x}{y}{u}}{t_2}}}
{\Nsequent{\NUPRL{H}}{\NUPRLequality{T}{\NUPRLsubb{u}{x}{s}{y}{t}}{t_2}}}
\end{array}\end{scriptsize}$$

\intitleb{Disjoint Unions}
The following rules are the disjoint union-elimination, disjoint
union-introduction (left and right), type equality for disjoint
unions, injection-introduction (left and right), and decide-computation
(left and right) rules, respectively:
$$\begin{scriptsize}\begin{array}{l}\infer[]
{\Nsequentext{\NUPRL{H},\NUPRLhyp{d}{\Nnoreadwritemod{\NUPRLunion{A}{B}}},\NUPRL{J}}{C}{\NUPRLdecide{d}{x}{t}{y}{u}}}
{
  \begin{array}{l}
    \Nsequentext{\NUPRL{H},\NUPRLhyp{d}{\Nnoreadwritemod{\NUPRLunion{A}{B}}},\NUPRLhyp{x}{A},\NUPRLsuba{\NUPRL{J}}{d}{\NUPRLinl{x}}}{\NUPRLsuba{C}{d}{\NUPRLinl{x}}}{t}
    \hspace{0.2in}
    \Nsequentext{\NUPRL{H},\NUPRLhyp{d}{\Nnoreadwritemod{\NUPRLunion{A}{B}}},\NUPRLhyp{y}{B},\NUPRLsuba{\NUPRL{J}}{d}{\NUPRLinr{y}}}{\NUPRLsuba{C}{d}{\NUPRLinr{y}}}{u}
  \end{array}
}
\\ \\
\infer[]
{\Nsequentext{\NUPRL{H}}{\NUPRLunion{A}{B}}{\NUPRLinl{a}}}
{
  \Nsequentext{\NUPRL{H}}{A}{a}
  &
  \Nsequent{\NUPRL{H}}{\NUPRLmember{\NUPRLuniverse{i}}{B}}
}
\hspace{0.2in}
\infer[]
{\Nsequentext{\NUPRL{H}}{\NUPRLunion{A}{B}}{\NUPRLinr{b}}}
{
  \Nsequentext{\NUPRL{H}}{B}{b}
  &
  \Nsequent{\NUPRL{H}}{\NUPRLmember{\NUPRLuniverse{i}}{A}}
}
\hspace{0.2in}
\infer[]
{\Nsequent{\NUPRL{H}}{\NUPRLequality{\NUPRLuniverse{i}}{\NUPRLunion{A_1}{B_1}}{\NUPRLunion{A_2}{B_2}}}}
{
  \Nsequent{\NUPRL{H}}{\NUPRLequality{\NUPRLuniverse{i}}{A_1}{A_2}}
  &
  \Nsequent{\NUPRL{H}}{\NUPRLequality{\NUPRLuniverse{i}}{B_1}{B_2}}
}
\\ \\
\infer[]
{\Nsequent{\NUPRL{H}}{\NUPRLequality{\NUPRLunion{A}{B}}{\NUPRLinl{a_1}}{\NUPRLinl{a_2}}}}
{
  \Nsequent{\NUPRL{H}}{\NUPRLequality{A}{a_1}{a_2}}
  &
  \Nsequent{\NUPRL{H}}{\NUPRLmember{\NUPRLuniverse{i}}{B}}
}
\hspace{0.2in}
\infer[]
{\Nsequent{\NUPRL{H}}{\NUPRLequality{\NUPRLunion{A}{B}}{\NUPRLinr{b_1}}{\NUPRLinr{b_2}}}}
{
  \Nsequent{\NUPRL{H}}{\NUPRLequality{B}{b_1}{b_2}}
  &
  \Nsequent{\NUPRL{H}}{\NUPRLmember{\NUPRLuniverse{i}}{A}}
}
\\ \\
\infer[]
{\Nsequent{\NUPRL{H}}{\NUPRLequality{T}{(\NUPRLdecide{\NUPRLinl{s}}{x}{t}{y}{u})}{t_2}}}
{\Nsequent{\NUPRL{H}}{\NUPRLequality{T}{\NUPRLsuba{t}{x}{s}}{t_2}}}
\hspace{0.2in}
\infer[]
{\Nsequent{\NUPRL{H}}{\NUPRLequality{T}{(\NUPRLdecide{\NUPRLinr{s}}{x}{t}{y}{u})}{t_2}}}
{\Nsequent{\NUPRL{H}}{\NUPRLequality{T}{\NUPRLsuba{u}{y}{s}}{t_2}}}
\end{array}\end{scriptsize}$$

\intitleb{Equality}
The following rules are the standard equality-introduction rule,
equality-elimination rule,
hypothesis rule, symmetry and
transitivity rules, respectively.
$$\begin{scriptsize}\begin{array}{l}\infer[]
{\Nsequent{\NUPRL{H}}{\NUPRLequality{\NUPRLuniverse{i}}{\NUPRLequalityb{A}{a_1}{a_2}}{\NUPRLequalityb{B}{b_1}{b_2}}}}
{
  \Nsequent{\NUPRL{H}}{\NUPRLequality{\NUPRLuniverse{i}}{A}{B}}
  &
  \Nsequent{\NUPRL{H}}{\NUPRLequality{A}{a_1}{b_1}}
  &
  \Nsequent{\NUPRL{H}}{\NUPRLequality{A}{a_2}{b_2}}
}
\hspace{0.2in}
\infer[]
{\Nsequentext{\NUPRL{H},\NUPRLhyp{z}{\NUPRLequality{A}{a}{b}},\NUPRL{J}}{C}{e}}
{
  \Nsequentext{\NUPRL{H},\NUPRLhyp{z}{\NUPRLequality{A}{a}{b}},\NUPRLsuba{\NUPRL{J}}{z}{\Naxiom}}{\NUPRLsuba{C}{z}{\Naxiom}}{e}
}
\\ \\
\infer[]
{\Nsequent{\NUPRL{H},\NUPRLhyp{x}{A},\NUPRL{J}}{\NUPRLmember{A}{x}}}
{}
\hspace{0.2in}
\infer[]
{\Nsequent{\NUPRL{H}}{\NUPRLequality{T}{a}{b}}}
{\Nsequent{\NUPRL{H}}{\NUPRLequality{T}{b}{a}}}
\hspace{0.2in}
\infer[]
{\Nsequent{\NUPRL{H}}{\NUPRLequality{T}{a}{b}}}
{
  \Nsequent{\NUPRL{H}}{\NUPRLequality{T}{a}{c}}
  &
  \Nsequent{\NUPRL{H}}{\NUPRLequality{T}{c}{b}}
}\end{array}\end{scriptsize}$$

The following rules allow fixing the realizer of a sequent, and
rewriting with an equality in an hypothesis, respectively:
$$\begin{scriptsize}\begin{array}{l}\infer[]
{\Nsequent{\NUPRL{H}}{\NUPRLmember{T}{t}}}
{\Nsequentext{\NUPRL{H}}{T}{t}}
\hspace{0.2in}
\infer[]
{\Nsequentext{\NUPRL{H},\NUPRLhyp{x}{A},\NUPRL{J}}{C}{t}}
{
  \Nsequentext{\NUPRL{H},\NUPRLhyp{x}{B},\NUPRL{J}}{C}{t}
  &
  \Nsequent{\NUPRL{H}}{\NUPRLequality{\NUPRLuniverse{i}}{A}{B}}
}\end{array}\end{scriptsize}$$

\intitleb{Universes}
Let $i$ be a lower universe than $j$.
The following rules are the standard universe-introduction
rule and the universe cumulativity rule, respectively.
$$\begin{scriptsize}\infer[]
{\Nsequent{\NUPRL{H}}{\NUPRLequality{\NUPRLuniverse{j}}{\NUPRLuniverse{i}}{\NUPRLuniverse{i}}}}
{}
\qquad\qquad\qquad
\infer[]
{\Nsequent{\NUPRL{H}}{\NUPRLmember{\NUPRLuniverse{j}}{T}}}
{\Nsequent{\NUPRL{H}}{\NUPRLmember{\NUPRLuniverse{i}}{T}}}
\end{scriptsize}$$

\intitleb{Sets}
The following rule is the standard set-elimination rule:
$$\begin{scriptsize}\infer[]
{\Nsequentext{\NUPRL{H},\NUPRLhyp{z}{\NUPRLset{x}{A}{B}},\NUPRL{J}}{C}{\NUPRLsuba{e}{a}{z}}}
{\Nsequentext{\NUPRL{H},\NUPRLhyp{z}{\NUPRLset{x}{A}{B}},\NUPRLhyp{a}{A},\NUPRLhhyp{b}{\NUPRLsuba{B}{x}{a}},\NUPRLsuba{\NUPRL{J}}{z}{a}}{\NUPRLsuba{C}{z}{a}}{e}}\end{scriptsize}$$
Note that we have used a new construct in the above rule: the
\emph{hidden} hypothesis $\NUPRLhhyp{b}{\NUPRLsuba{B}{x}{a}}$. The
main feature of hidden hypotheses is that their names cannot occur in
realizers (which is why we ``box'' those hypotheses). Intuitively, this
is because the proof that $B$ is true is discarded in the proof that
the set type $\NUPRLset{x}{A}{B}$ is true and therefore cannot occur
in computations.
Hidden hypotheses can be unhidden using the following rule:
$$\begin{scriptsize}\infer[]
{\Nsequentext{\NUPRL{H},\NUPRLhhyp{x}{T},\NUPRL{J}}{\NUPRLequality{A}{a}{b}}{\Naxiom}}
{\Nsequentext{\NUPRL{H},\NUPRLhyp{x}{T},\NUPRL{J}}{\NUPRLequality{A}{a}{b}}{\Naxiom}}\end{scriptsize}$$
which is valid since the realizer is $\Naxiom$ and therefore does not
make use of $x$.

The following rules are the standard set-introduction rule, type
equality for the set type, and introduction rule for members of set
types, respectively.
$$\begin{scriptsize}\begin{array}{l}\infer[]
{\Nsequentext{\NUPRL{H}}{\NUPRLset{x}{A}{B}}{a}}
{
  \Nsequent{\NUPRL{H}}{\NUPRLmember{A}{a}}
  &
  \Nsequent{\NUPRL{H}}{\NUPRLsuba{B}{x}{a}}
  &
  \Nsequent{\NUPRL{H},\NUPRLhyp{z}{A}}{\NUPRLmember{\NUPRLuniverse{i}}{\NUPRLsuba{B}{x}{z}}}
}
\hspace{0.1in}
\infer[]
{\Nsequent{\NUPRL{H}}{\NUPRLequality{\NUPRLuniverse{i}}{\NUPRLset{x_1}{A_1}{B_1}}{\NUPRLset{x_2}{A_2}{B_2}}}}
{
  \Nsequent{\NUPRL{H}}{\NUPRLequality{\NUPRLuniverse{i}}{A_1}{A_2}}
  &
  \Nsequent{\NUPRL{H},\NUPRLhyp{y}{A_1}}{\NUPRLequality{\NUPRLuniverse{i}}{\NUPRLsuba{B_1}{x_1}{y}}{\NUPRLsuba{B_2}{x_2}{y}}}
}
\\ \\
\infer[]
{\Nsequent{\NUPRL{H}}{\NUPRLequality{\NUPRLset{x}{A}{B}}{a}{b}}}
{
  \Nsequent{\NUPRL{H},\NUPRLhyp{z}{A}}{\NUPRLmember{\NUPRLuniverse{i}}{\NUPRLsuba{B}{x}{z}}}
  &
  \Nsequent{\NUPRL{H}}{\NUPRLequality{A}{a}{b}}
  &
  \Nsequent{\NUPRL{H}}{\NUPRLsuba{B}{x}{a}}
}\end{array}\end{scriptsize}$$

\intitleb{Intersection}
The following rules are the intersection elimination and introduction
rules:
$$\begin{scriptsize}\begin{array}{l}
\infer[]
{\Nsequentext{\NUPRL{H},\NUPRLhyp{x}{\Nisect{A}{B}},\NUPRL{J}}{C}{\NUPRLsuba{\NUPRLsuba{t}{y}{x}}{z}{\Naxiom}}}
{
  \Nsequentext{\NUPRL{H},\NUPRLhyp{x}{\Nisect{A}{B}},\NUPRLhyp{y}{A},\NUPRLhyp{z}{\NUPRLequality{A}{y}{x}},\NUPRLsuba{\NUPRL{J}}{x}{y}}{\NUPRLsuba{C}{x}{y}}{t}
}
\hspace{0.2in}
\infer[]
{\Nsequentext{\NUPRL{H},\NUPRLhyp{x}{\Nisect{A}{B}},\NUPRL{J}}{C}{\NUPRLsuba{\NUPRLsuba{t}{y}{x}}{z}{\Naxiom}}}
{
  \Nsequentext{\NUPRL{H},\NUPRLhyp{x}{\Nisect{A}{B}},\NUPRLhyp{y}{B},\NUPRLhyp{z}{\NUPRLequality{B}{y}{x}},\NUPRLsuba{\NUPRL{J}}{x}{y}}{\NUPRLsuba{C}{x}{y}}{t}
}
\\ \\
\infer[]
{\Nsequentext{\NUPRL{H}}{\Nisect{A}{B}}{t}}
{
  \Nsequentext{\NUPRL{H}}{A}{t}
  &
  \Nsequentext{\NUPRL{H}}{b}{t}
}
\hspace{0.2in}
\infer[]
{\Nsequent{\NUPRL{H}}{\NUPRLequality{\NUPRLuniverse{i}}{\Nisect{A_1}{B_1}}{\Nisect{A_2}{B_2}}}}
{
  \Nsequent{\NUPRL{H}}{\NUPRLequality{\NUPRLuniverse{i}}{A_1}{A_2}}
  &
  \Nsequent{\NUPRL{H}}{\NUPRLequality{\NUPRLuniverse{i}}{B_1}{B_2}}
}
\hspace{0.2in}
\infer[]
{\Nsequent{\NUPRL{H}}{\NUPRLequality{\Nisect{A}{B}}{t_1}{t_2}}}
{
  \Nsequent{\NUPRL{H}}{\NUPRLequality{A}{t_1}{t_2}}
  &
  \Nsequent{\NUPRL{H}}{\NUPRLequality{B}{t_1}{t_2}}
}
\end{array}\end{scriptsize}$$

\intitleb{Subsingleton}
The following rules are the subsingleton elimination and introduction
rules:%
$$\begin{scriptsize}\begin{array}{l}\infer[]
{\Nsequent{\NUPRL{H},\NUPRLhyp{x}{\Nqsquash{A}},\NUPRL{J}}{\NUPRLequality{T}{t}{u}}}
{
  \Nsequent{\NUPRL{H},\NUPRLhyp{x}{\Nqsquash{A}},\NUPRL{J}}{\NUPRLmember{\NUPRLuniverse{i}}{T}}
  &
  \Nsequent{\NUPRL{H},\NUPRLhyp{x}{\Nqsquash{A}},\NUPRL{J},\NUPRLhyp{y}{A},\NUPRLhyp{z}{A}}{\NUPRLequality{T}{\NUPRLsuba{t}{x}{y}}{\NUPRLsuba{u}{x}{z}}}
}
\\ \\
\infer[]
{\Nsequentext{\NUPRL{H}}{\Nqsquash{A}}{t}}
{
  \Nsequentext{\NUPRL{H}}{A}{t}
}
\hspace{0.2in}
\infer[]
{\Nsequent{\NUPRL{H}}{\NUPRLequality{\NUPRLuniverse{i}}{\Nqsquash{A_1}}{\Nqsquash{A_2}}}}
{
  \Nsequent{\NUPRL{H}}{\NUPRLequality{\NUPRLuniverse{i}}{A_1}{A_2}}
}
\hspace{0.2in}
\infer[]
{\Nsequent{\NUPRL{H}}{\NUPRLequality{\Nqsquash{A}}{t_1}{t_2}}}
{
  \Nsequent{\NUPRL{H}}{\NUPRLmember{A}{t_1}}
  &
  \Nsequent{\NUPRL{H}}{\NUPRLmember{B}{t_2}}
}
\end{array}\end{scriptsize}$$

\intitleb{Natural Numbers}
The following rules are the rules for $\NUPRLnat$:
$$\begin{scriptsize}\begin{array}{l}\infer[]
{\Nsequentext{\NUPRL{H},\NUPRLhyp{x}{\Nnoreadwritemod{\NUPRLnat}},\NUPRL{J}}{C}{\Nnatrec{x}{b}{\NUPRLlam{y}{\NUPRLlam{r}{s}}}}}
{
  \Nsequentext{\NUPRL{H},\NUPRLhyp{x}{\Nnoreadwritemod{\NUPRLnat}},\NUPRL{J}}{\NUPRLsuba{C}{x}{\metanat{0}}}{b}
  &
  \Nsequentext{\NUPRL{H},\NUPRLhyp{x}{\Nnoreadwritemod{\NUPRLnat}},\NUPRL{J},\NUPRLhyp{y}{\Nnoreadwritemod{\NUPRLnat}},\NUPRLhyp{r}{\NUPRLsuba{C}{x}{y}}}{\NUPRLsuba{C}{x}{\Nsucc{y}}}{s}
}
\\ \\
\infer[]
{\Nsequentext{\NUPRL{H}}{\NUPRLnat}{\metanat{n}}}
{}
\hspace{0.2in}
\infer[]
{\Nsequent{\NUPRL{H}}{\NUPRLequality{\NUPRLuniverse{i}}{\NUPRLnat}{\NUPRLnat}}}
{}
\hspace{0.2in}
\infer[]
{\Nsequent{\NUPRL{H}}{\NUPRLequality{\NUPRLnat}{\metanat{n}}{\metanat{n}}}}
{}
\hspace{0.2in}
\infer[]
{\Nsequent{\NUPRL{H}}{\NUPRLequality{\NUPRLnat}{\Nsucc{t}}{\Nsucc{u}}}}
{
  \Nsequent{\NUPRL{H}}{\NUPRLequality{\NUPRLnat}{t}{u}}
}
\end{array}\end{scriptsize}$$

\intitleb{Effect Restrictions}
The following rules are the rules for $\NnowriteSYMB$, $\NnoreadSYMB$,
$\Npure$:%
$$\begin{scriptsize}\begin{array}{l}
\infer[]
{\Nsequentext{\NUPRL{H}}{\NnowriteSYMB}{t}}
{
  \Nnoreads{t}
}
\hspace{0.2in}
\infer[]
{\Nsequentext{\NUPRL{H}}{\NnoreadSYMB}{t}}
{
  \Nnowrites{t}
}
\hspace{0.2in}
\infer[]
{\Nsequentext{\NUPRL{H}}{\Npure}{t}}
{
  \Nnonames{t}
}
\\ \\
\infer[]
{\Nsequent{\NUPRL{H}}{\NUPRLequality{\NnowriteSYMB}{t_1}{t_2}}}
{
  \Nnowrites{t_1}
  &
  \Nnowrites{t_2}
}
\hspace{0.2in}
\infer[]
{\Nsequent{\NUPRL{H}}{\NUPRLequality{\NnoreadSYMB}{t_1}{t_2}}}
{
  \Nnoreads{t_1}
  &
  \Nnoreads{t_2}
}
\hspace{0.2in}
\infer[]
{\Nsequent{\NUPRL{H}}{\NUPRLequality{\Npure}{t_1}{t_2}}}
{
  \Nnonames{t_1}
  &
  \Nnonames{t_2}
}
\\ \\
\infer[]
{\Nsequent{\NUPRL{H}}{\NUPRLequality{\NUPRLuniverse{i}}{\NnowriteSYMB}{\NnowriteSYMB}}}
{}
\hspace{0.2in}
\infer[]
{\Nsequent{\NUPRL{H}}{\NUPRLequality{\NUPRLuniverse{i}}{\NnoreadSYMB}{\NnoreadSYMB}}}
{}
\hspace{0.2in}
\infer[]
{\Nsequent{\NUPRL{H}}{\NUPRLequality{\NUPRLuniverse{i}}{\Npure}{\Npure}}}
{}
\end{array}\end{scriptsize}$$
Note that all these rules require syntactic checks, where
$\Nnonames{t}$ states that no expression of the form~$\Mcn$ or
$\Nfresh{x}{t_1}$ occurs in $t$;
$\Nnowrites{t}$ states that no expression of the form
$\Nchoose{t_1}{t_2}$ or $\Nfresh{x}{t_1}$ occurs in $t$;
and $\Nnoreads{t}$ states that no expression of the form $\Nread{t_1}$
occurs in $t$.

\else%
\fi

\end{document}